\documentclass[review,onefignum,onetabnum]{siamart171218}
\usepackage{multirow}
\usepackage{amsmath}
\usepackage{color}
\usepackage{upgreek}

\usepackage{amssymb}
\usepackage{mathrsfs}
\def\Q{\mathbf{Q}}
\def\P{\mathbf{P}}
\def\I{\mathbf{I}}

\def\n{\mathbf{n}}

\def\x{\mathbf{x}}

\def\_v{\mathbf{v}}

\def\S{\mathbf{S}}

\usepackage{lipsum}
\usepackage{amsfonts}
\usepackage{graphicx}
\usepackage{epstopdf}
\usepackage{algorithmic}
\newsiamremark{remark}{Remark}
\newsiamremark{hypothesis}{Hypothesis}
\crefname{hypothesis}{Hypothesis}{Hypotheses}
\newsiamthm{claim}{Claim}

\newcommand{\TheTitle}{Smectic Liquid Crystal in Modified LdG Model}
\newcommand{\TheAuthors}{B. Shi, Y. Han, C. Ma, A. Majumdar, L. Zhang}

\headers{\TheTitle}{\TheAuthors}
\title{{\TheTitle}
\thanks{LZ is supported by the National Natural Science Foundation of China (grants 12225102, T2321001, 12288101).
AM is supported by the University of Strathclyde New Professors Fund, 
a Friedrich Wilhelm Bessel Research Award from the Humboldt Foundation, a Leverhulme Research Project Grant RPG-2021-401, and the INI network grant supported by EP/R014604/1. This work was partially supported by Royal Society Newton Advanced Fellowship NAF/R1/180178 awarded to AM and LZ.  YH is supported by a Leverhulme Research Project Grant RPG-2021-401.}
}

\author{Baoming Shi\thanks{School of Mathematical Sciences, Peking University, Beijing 100871, China (\email{ming123@stu.pku.edu.cn}).}
\and  Yucen Han\thanks{Department of Mathematics and Statistics, University of Strathclyde, G1 1XQ, UK (\email{yucen.han@strath.ac.uk}).}
\and  Chengdi Ma\thanks{School of Mathematical Sciences, Peking University, Beijing 100871, China (\email{mcd2020@stu.pku.edu.cn}).}
\and  Apala Majumdar\thanks{Department of Mathematics and Statistics, University of Strathclyde, G1 1XQ, UK (\email{apala.majumdar@strath.ac.uk}).}
\and Lei Zhang\thanks{Beijing International Center for Mathematical Research, Center for Quantitative Biology, Center for Machine Learning Research, Peking University, Beijing 100871, China (\email{zhangl@math.pku.edu.cn}).}
}
\ifpdf
\hypersetup{
  pdftitle={\TheTitle},
  pdfauthor={\TheAuthors}
}
\fi
\renewcommand{\TheTitle}{A Modified Landau-de Gennes Theory for Smectic Liquid Crystals: Phase Transitions and Structural Transitions}
\begin{document}

\maketitle
\textcolor{red}{This work has been published in SIAM Journal on Applied Mathematics, 2025, 85(2): 821-847.
Please refer to the official publication for citation. The codes can be obtained
at https://github.com/BaomingShi/searching-minimizer-mLdG-model.
}
\\
\begin{abstract}
We mathematically model Smectic-A (SmA) phases with a modified Landau-de Gennes (mLdG) model as proposed in \cite{xia2021structural}. The orientational order of the SmA phase is described by a tensor-order parameter $\Q$, and the positional order is described by a real scalar $u$, which models the deviation from the average density of liquid crystal molecules. 
Firstly, we prove the existence and regularity of global minimisers of the mLdG free energy in three-dimensional settings. Then, we analytically prove that the mLdG model can capture the Isotropic-Nematic-Smectic phase transition as a function of temperature, under some assumptions. 
Further, we explore stable smectic phases on a square domain, with edge length $\lambda$, and tangent boundary conditions. 
We use heuristic arguments to show that defects repel smectic layers and that nematic ordering promotes layer formation. We use asymptotic arguments in the $\lambda\to 0$ and $\lambda\to\infty$ limits which reveal the correlation between the number and thickness of smectic layers, the amplitude of density fluctuations with the phenomenological parameters in the mLdG energy. 
For finite values of $\lambda$, we numerically recover BD-like and D-like stable smectic states observed in experiments. 
We also study the frustrated mLdG energy landscape and give  numerical examples of transition pathways between distinct mLdG energy minimisers. 
\end{abstract}



\section{Introduction}
Liquid crystals are mesophases intermediate between the solid and liquid states, characterized by orderly molecular arrangements \cite{de1993physics}, that is, the constituent molecules align along
certain locally preferred directions, referred to as ``directors" in the literature. These orderly molecular arrangements give rise to distinctive optical and electrical properties in liquid crystals, making them valuable in display technologies, optical devices, and sensors \cite{LAGERWALL20121387,bisoyi2021liquid,edwards2020interaction}. Liquid crystals can exhibit different phases, such as nematic and smectic phases. The nematic phase has long-range orientational order but lacks positional order, while the smectic phase possesses both long-range orientational order and partial positional order, leading to a layered structure with positional coherence within the layers \cite{de1972analogy}. There are  several smectic phases, such as Smectic-A and Smectic-C, each with distinct characteristics \cite{han2015microscopic}. In the Smectic-A phase, the director is parallel to the normal of the layer. In contrast, in the Smectic-C phase, there is a non-zero angle between the director and the normal of the layer. In this paper, we focus on the Smectic-A phase, which will simply be referred to  as the ``smectic phase''.

External constraints, such as confinement and boundary anchoring, can induce deformations in the liquid crystal. These deformations may not coincide with the liquid crystal phase in the bulk, leading to geometric frustrations. As a result, a diverse array of textures with characteristic defect structures may spontaneously assemble \cite{huang2022defect, williams1975dislocations}. For instance, when a smectic liquid crystal is deposited on a substrate that promotes varying boundary anchoring, their layers may bend and form focal conic domains (FCDs) \cite{kleman2000grain}. These FCDs have been utilized as guides for colloidal dispersion \cite{milette2012reversible}, in soft lithography \cite{kim2010self}, and as templates for superhydrophobic surfaces \cite{kim2014creation}. The experimental observations in \cite{cortes2016colloidal} suggest a stable BD-type smectic profile on square domains, with a pair of line defects localised near a pair of opposite square edges. 
This BD-type configuration is stablised by the smectic positional ordering but is 
unstable for pure nematics \cite{tsakonas2007multistable,yin2020construction,shi2023hierarchies}, which indicates the distinctive properties of confined smectic configurations.

Recent work has focused on the Nematic-Smectic (N-S) phase transition and the coupling between directors and smectic layers. For example, the existence of geometric memory in the N-S transition leads to FCDs melting into a dense array of boojums defects \cite{suh2019topological}. By using colloidal silica rods and leveraging their significant density difference with the dispersing solvent, nematic and smectic phases can be confined within a single chamber which produces a smectic-nematic interface, and the directors in the smectic-nematic interface leave fingerprints in the nematic slice \cite{cortes2016colloidal}. On spherical shells, the N-S phase transition, or the emergence of the layer structure, initially occurs on the thicker side of the shell, distant from the point defects \cite{liang2011nematic}. These experimental findings inspire us to mathematically study the N-S phase transition and structural phase transitions in confined smectic systems.

The very complicated structures that emerge in the frustrated smectic phase are challenging to model mathematically. The key point in modelling the smectic phase is to incorporate the nematic director and the layer structure, i.e. an additional positional order parameter must be introduced to describe the modulation of the density of liquid crystal molecules, compared to a simple nematic phase. In recent decades, several powerful continuum mathematical theories have been used for the nematic phase such as the microscopic Onsager model, the macroscopic Landau-de Gennes (LdG) model, the macroscopic Oseen-Frank model, and the Ericksen-Leslie model \cite{wang2021modeling}. For modelling the smectic phase, an additional positional order parameter is required to construct the layered structure. For instance, the extended Maier-Saupe model \cite{mcmillan1971simple} is a molecular model for the smectic phase, which qualitatively predicts the N-S phase transition as a function of temperature. The molecular model is computationally expensive due to its inherent high-dimensional complexity but its parameters can be correlated with the underpinning molecular structures. For computational convenience, there have been competing phenomenological theories for smectic phases, obtained by adding the density modulation to the Oseen-Frank energy or the LdG energy for nematic phases \cite{chen1976landau,pevnyi2014modeling,ball2015discontinuous,xia2021structural,izzo2020landau,biscari2007landau}. These phenomenological models can successfully predict the structures observed in experiments, although there is no mapping between the phenomenological parameters and structural details e.g., properties of smectic layers. Most of the existing work focuses on numerical results, with a lack of interpretability of the models, which is essential for studying structural phase transitions and for controlling properties of confined smectic systems. 

We address some of these questions by a systematic study of the modified Landau-de Gennes (mLdG) model as presented in \cite{xia2021structural}, which is adept at capturing geometric frustration, FCDs, and oily streaks \cite{michel2006structure}, commonly observed in experiments on confined smectics. In the mLdG framework, there are two order parameters: the LdG nematic order parameter, $\Q$, described in detail in the next sections that encodes the (nematic) directors and a positional order parameter $u$, which models the deviation of the molecular density from the average density. The mLdG free energy comprises the LdG free energy (which depends on $\Q$ and its gradient), a bulk smectic energy pivotal for nematic-smectic transitions, and a nematic-smectic coupling energy.  The bulk smectic energy is a standard Ginzburg-Landau potential or a quartic polynomial in $u$. The nematic-smectic coupling energy depends on the Hessian of $u$ and is parameterised by two phenomenological parameters - a coupling strength, $B_0$, and a second parameter, $q$, which determines the thickness and multiplicity of smectic layers, at least in the mLdG framework. The nematic-smectic coupling energy  determines the relative alignment of the layer normals and directors, and within the remit of our work, the mLdG energy minimisers have co-aligned layer normals and directors and are hence, thought to model the SmA phase.  The mLdG energy minimisers model stable and experimentally relevant (observable) smectic phases. 


Firstly, in Section~\ref{sec: model}, we prove the existence and regularity of the mLdG energy minimiser, in three-dimensional settings, subject to strong and weak versions of experimentally relevant tangent boundary conditions for the directors. The tangent boundary conditions require the directors to be tangent to (in the plane of) the boundary. In Section~\ref{phase transition}, we analytically study the Isotropic-Nematic-Smectic phase transitions as a function of temperature. The LdG bulk energy has analytic critical points: isotropic and nematic critical points under periodic boundary conditions \cite{de1993physics}. We cannot find analytic critical points of the mLdG energy easily and this poses technical challenges in demonstrating that the mLdG energy can capture the N-S phase transition. We show that there are two critical temperatures, $T_1^* > T_2^*$ in the mLdG model. Then we prove that the nematic phase loses stability at $T<T_2^*$ by studying the second variation of the mLdG energy and provide an analytic estimate of the Morse index of the nematic critical point for $T<T_2^*$. We further use the Crandall and
Rabinowitz bifurcation theorem \cite{crandall1971bifurcation} to demonstrate that the nematic phase undergoes a pitchfork bifurcation at $T=T_2^*$, accompanied by the appearance of layered smectic structures. In Section~\ref{Sec: confinement}, we demonstrate that the nematic-smectic coupling term favours the formation of layered structures in regions of strong nematic or orientational ordering, again something which could be experimentally checked. 
Lastly, we study mLdG energy minimisers on square domains, as a function of the temperature and square edge length $\lambda$, subject to tangent boundary conditions for the directors on the edges. We draw on parallels with the nematic study in \cite{han2020reduced, luo2012multistability, yin2020constrained}, provide some physical interpretations of the phenomenological parameters in the bulk smectic energy and the nematic-smectic coupling energy and also give some numerical examples of transition pathways between distinct energy minimisers. The energy landscape is very frustrated with multiple minimisers, that have subtle differences in their structural properties, and this introduces new challenges in the study of mLdG solution landscapes. We conclude with some perspectives in Section~\ref{conclusion}. 

\section{Theoretical framework}\label{sec: model}
The Landau-de Gennes (LdG) model \cite{de1993physics} is the most celebrated continuum theory for nematic liquid crystals and has been hugely successful for describing the Isotropic-Nematic (I-N) phase transition and structural transitions for nematics \cite{majumdar2010equilibrium}. The LdG theory describes the nematic phase by the LdG $\Q$-tensor order parameter, which is a traceless and symmetric $3\times 3$ matrix. The $\Q$ tensor is isotropic if $\Q=0$, uniaxial if $\Q$ has a pair of degenerate nonzero eigenvalues, and biaxial if $\Q$ has three distinct eigenvalues \cite{de1993physics}. A uniaxial nematic phase has a single distinguished direction of averaged molecular alignment, modelled by the eigenvector with the non-degenerate eigenvalue. A biaxial nematic phase has a primary and a secondary nematic director. In approximately two-dimensional (2D) scenarios, we can use the reduced Landau-de Gennes (rLdG) model, with the rLdG order parameter - a symmetric and traceless $2\times2$ matrix with only two degrees of freedom: one degree of freedom for the nematic director in the plane and the second degree of freedom describes the degree of ordering about the 2D director \cite{yin2020construction,han2020reduced,shi2022nematic,10.1093/imamat/hxad031}. In this paper, we use a modified LdG (mLdG) theory to study confined smectic phases, wherein we use either the LdG or the rLdG order parameter to describe the orientational/nematic ordering with an additional real-valued positional order parameter $u$  and additional energy terms to describe the intrinsic layering of smectic phases \cite{xia2021structural,xia2023variational}. 

\subsection{Preliminaries} \label{sec:prelim}
The modified Landau-de Gennes (mLdG) energy is proposed in \cite{xia2021structural,xia2023variational} and is given by
\begin{equation}
    E(\Q,u)=\int_{\Omega} \left(f_{LdG}(\Q) + f_{bs}(u)+ f_{int}(\Q,u)\right)\mathrm{d}\x,
    \label{energy}
\end{equation}
where $\Omega \subset \mathbb{R}^3$ is the working domain, the nematic order parameter $\Q(\x)\in \mathbb{R}^{3\times 3}$, and the positional order parameter, $u(\x)\in\mathbb{R}$, models the deviation of the molecular density from the average molecular density at position $\x$. The positional order parameter, $u$, is the real part of the complex order parameter in \cite{de1972analogy}. For further details, we refer the reader to \cite{pevnyi2014modeling}. The first term in \eqref{energy} is the LdG free energy density,
\begin{equation}
    f_{LdG}(\Q): = \frac{K}{2}\left| \nabla \Q \right|^2 + f_{bn}\left( \Q \right), 
\end{equation}
where $K$ is a positive material-dependent elastic constant. The elastic energy density penalizes spatial inhomogeneities, and the thermotropic bulk energy density, $f_{bn}$, dictates the preferred nematic liquid crystal (NLC) phase as a function of temperature,
\begin{equation}\label{f_B}
    f_{bn}(\Q): = \frac{A}{2}\mathrm{tr} \Q^2 - \frac{B}{3} \mathrm{tr} \Q^3 + \frac{C}{4} (\mathrm{tr} \Q^2)^2,\\
\end{equation}
where $A=\alpha_1 (T - T_1^*)$ is the rescaled temperature, with $\alpha_1>0$, and $T_1^*$ is a characteristic liquid crystal temperature; $B, C>0$ are material-dependent bulk constants. For example, typical values for the representative NLC material MBBA are $B=0.64\times10^4 \text{Nm}^{-2}$, $C=0.35\times10^4 \text{Nm}^{-2}$ and $K=4\times10^{-11} \text{N}$ \cite{majumdar2010equilibrium, wojtowicz1975introduction}. The minimisers of $f_{bn}$ depend on $A$ and determine the NLC phase for spatially homogeneous samples. The minimiser of $f_{bn}$ is the isotropic state for $A> \frac{B^2}{27 C}$. For $A<\frac{B^2}{27 C}$, the minimisers of $f_{bn}$ constitute a continuum of $\Q$-tensors defined below:

\[
\mathcal{N} =  \left\{ \Q = s_+\left(\mathbf{n}\otimes \mathbf{n} - \frac{\mathbf{I_3}}{3} \right) \right\},s_+ = \frac{B + \sqrt{B^2 - 24 AC}}{4C},
\]
where $\mathbf{n}$ is an arbitrary unit vector field (referred to as the nematic director), and $\mathbf{I}_3$ is the $3\times 3$ identity matrix.

The second term in \eqref{energy} is the bulk energy density of the smectic order parameter $u$, which can be derived from the Landau theory of phase transitions \cite{de1972analogy,pevnyi2014modeling,izzo2020landau}:
\begin{equation}
f_{bs}(u)=\frac{a}{2}u^2 + \frac{b}{3} u^3 +  \frac{c}{4} u^4,
\end{equation}
where $a=\alpha_2(T-T^*_2)$ is a temperature-dependent parameter with $\alpha_2>0$, and $T_2^*<T_1^*$ is a critical material temperature related to N-S phase transition; 
$b,c>0$ are material-dependent constants. A non-zero $b$ will result in non-symmetrical layer structures \cite{pevnyi2014modeling}, and we take $b=0$ to study symmetric layer structures. When $a<0$, i.e. the temperature is low enough, the minimisers of $f_{bs}(u)$ prefer a non-zero density distribution, $u$. 

The third term in \eqref{energy} is the coupling term between the smectic and nematic order parameters:
\begin{equation}\label{eq: coupling term}
f_{int}(\Q,u) =
\begin{cases}
B_0 \left|D^2u\right|^2, A\geqslant \frac{B^2}{24C},\\
B_0 \left|D^2u+q^2\left(\frac{\Q}{s_+}+\frac{\mathbf{I}_3}{3}\right)u\right|^2, otherwise,
\end{cases}
\end{equation}
where $B_0$ is a phenomenological coupling constant between $\Q$ and $u$, typically chosen to be on the scale of \( 1/q^4 \) to counterbalance the magnitude of the coupling energy density \( f_{\text{int}}\) \cite{xia2021structural,pevnyi2014modeling}. $D^2 u$ is the Hessian of $u$. For a low temperature $T<T_2^*$ and assuming a uniaxial $\Q = s_+\left(\n\otimes\n-\I_3/3\right)$, $f_{int}$ is minimized by $u=sin(q\n\cdot \x)$, which corresponds to a layered structure that has the layer normal co-aligned with the uniaxial director $\n$, characteristic of the SmA phase.  Consequently, $q$ is often identified with the wave number of the SmA layers \cite{xia2021structural,pevnyi2014modeling}, and is expected to be related to the SmA layer thickness $l$ by $q = 2\pi/l$. The layer thickness $l$ of a homogeneous SmA,
is usually slightly larger than the long axis of a rod-like liquid crystal molecule, $L$, but less than $2L$ \cite{mei2015molecular}. The layer thickness of the equal mass mixture of 8OPhPy8 and 6OPhPy8 in the SmA phase is about $28.5$ Angstrom in \cite{Enz2013Electrospun}. 

The admissible $\Q$-tensors belong to the space 
\begin{equation}
W^{1,2}_{\Omega,\S_0}=\left\{\Q \in \S_0|\Q \in W_\Omega^{1,2}\right\},
\end{equation}
and the admissible smectic order parameter, $u$, belongs to $W^{2,2}_\Omega$,
where 
\begin{equation}
\begin{aligned}
\S_0:=&\left\{\Q \in \mathbb{R}^{3\times 3}: \Q_{ij}=\Q_{ji},\sum_{i = 1}^3 \Q_{ii}=0\right\},\\
 W_\Omega^{k,p}=& \left\{ u:\Omega \to \mathbb{R}:  \int_\Omega \left(|u|^p+\sum_{|\alpha|\leqslant k}|D^\alpha u|^p \right)\mathrm {d}x<\infty \right\}.
\end{aligned}
\end{equation}

To study the Isotropic-Nematic-Smectic (I-N-S) phase transition and structural transitions in confinement, we consider three different kinds of boundary conditions: 
(1) Periodic boundary condition for $\Q$ and $u$ on a  one-dimensional domain $\Omega=[0,h]$:
\begin{equation}
\begin{cases}
    \Q(0)=\Q(h),D_x\Q(0)=D_x\Q(h), \\
    u(0)=u(h),D_x u(0)=D_x u(h),D_{xx} u(0)=D_{xx} u(h).
    \end{cases}
    \label{periodic bc}
\end{equation}
We impose periodic boundary conditions on the derivative of $\Q$ to ensure that $\Q$ is smooth at the boundaries. Similarly, we impose periodic boundary condition on the second derivative of $u$.

(2) Dirichlet boundary conditions for $\Q$ \cite{yin2020construction,han2020reduced} and natural boundary condition for $u$ are specified as follows,
\begin{equation}
\begin{cases}
\Q=\Q_{bc} \text{ on } \partial \Omega,\\
\left(D^2u+q^2\left(\frac{\mathbf{Q}}{s_+}+\frac{\I_3}{3}\right)u\right): \vec{\nu} \otimes \vec{\nu}=0, \text{ on } \partial \Omega, \\
\left[\nabla \cdot \left(D^2u+q^2\left(\frac{\Q}{s_+}+\frac{\I_3}{3}\right)u\right) \right]\cdot\vec{\nu}+\nabla \cdot \left[ \P_{\vec{\nu}} \left(D^2u+q^2\left(\frac{\mathbf{Q}}{s_+}+\frac{\I_3}{3}\right)u\right)\vec{\nu} \right]=0, \text{ on } \partial \Omega,
\end{cases}
\end{equation}
with the specified Dirichlet boundary $\Q_{bc}\in W^{\frac{1}{2},2}_{\partial \Omega,\S_0}$, where $W^{\frac{1}{2},2}_{\partial \Omega,\S_0}$ is a fractional order Sobolev space which is the image space of the trace operator on $W^{1,2}_{\Omega,\S_0}$ \cite{adams2003sobolev}. \(\vec{\nu}\) is the outer normal vector, and \(\mathbf{P}_{\vec{\nu}} = \mathbf{I}_3 - \vec{\nu} \otimes \vec{\nu}\) is the projection matrix onto the tangential plane of the boundary. One admissible example is the tangential Dirichlet boundary condition in \cite{10.1093/imamat/hxad031}, for which the nematic director is tangent to or in the plane of the domain boundary and such boundary conditions are motivated by experiments \cite{tsakonas2007multistable}. The natural boundary condition for $u$ implies that the molecular density distribution is unconstrained. 

(3) We can also use weak boundary conditions or surface energies for the LdG order parameter as shown below 
\cite{nobili1992disorientation}, and the total energy is 
\begin{equation}
  \Tilde{E}(\Q,u)= E(\Q,u)+\omega \int_{\partial \Omega} \Vert\Q-\Q_{bc}\Vert^2 \mathrm{d}S, \text{ }\omega\geqslant 0,
  \label{weak anchoring energy}
\end{equation}
where $\omega\geqslant 0$ is the penalty strength. From the method of variations, the critical point of \eqref{weak anchoring energy} satisfies the weak anchoring boundary conditions for $\Q$ \cite{shi2024multistability} and natural boundary condition for $u$,
$$
\begin{cases}
\frac{\partial{\Q}}{\partial \vec{\nu}}+\frac{2\omega}{K}\left(\Q-\Q_{bc}\right)=0 \text{ on } \partial \Omega\\
\left(D^2u+q^2\left(\frac{\mathbf{Q}}{s_+}+\frac{\I_3}{3}\right)u\right): \vec{\nu} \otimes \vec{\nu}=0, \text{ on } \partial \Omega,\\
\left[\nabla \cdot \left(D^2u+q^2\left(\frac{\Q}{s_+}+\frac{\I_3}{3}\right)u\right) \right]\cdot\vec{\nu}+\nabla \cdot \left[ \P_{\vec{\nu}} \left(D^2u+q^2\left(\frac{\mathbf{Q}}{s_+}+\frac{\I_3}{3}\right)u\right)\vec{\nu} \right]=0, \text{ on } \partial \Omega.
\end{cases}
$$

\subsection{The proofs of existence and regularity}

\begin{proposition}
The mLdG energy functional \eqref{energy} has at least a global minimiser $(\Tilde{\Q},\Tilde{u})$ in $W^{1,2}_{\Omega,\S_0}\times W^{2,2}_\Omega$, subject to the above three types of boundary conditions. 
\end{proposition}

\begin{proof}
The admissible space $W^{1,2}_{\Omega,\S_0}\times W^{2,2}_\Omega$ is non-empty. The existence of a global minimiser of \eqref{energy} under Dirichlet boundary conditions for both $\Q$ and $u$ has been proven in \cite{xia2023variational}. We prove that the existence of a global minimiser also holds with weak anchoring for $\Q$ and natural boundary condition for $u$. The bulk energy $f_{bn}(\Q)$ is a fourth-order polynomial of $\Q$, and the fourth-order term is positive because $C>0$. Hence, there exists a positive $M$ (that depends on $A$, $B$, $C$) such that $f_{bn}(\Q)\geqslant \frac{C}{8}|\Q|^4$ for $|\Q|^2\geqslant M$, so that 
\begin{equation}
  f_{bn}(\Q) \geqslant \begin{cases}
 \frac{C}{8}|\Q|^4\geqslant \frac{MC}{8}|\Q|^2,|\Q|^2\geqslant M,\\         \min_{|\Q|^2\leqslant M} f_{bn}(\Q) =constant, |\Q|^2\leqslant M.
 \end{cases}
\end{equation}
Thus, there exist two positive constants, $C_1(A,B,C)>0,C_2(A,B,C)>0$, such that
\begin{equation}
 \int_\Omega f_{bn}(\Q)\mathrm{d}\x \geqslant C_1(A,B,C)\Vert \Q \Vert^2_{L^2_{\Omega,\S_0}}-C_2(A,B,C), 
\end{equation}
and
\begin{equation}
\begin{aligned}
 &\int_\Omega \frac{K}{2} \left|\nabla \Q \right|^2+f_{bn}(\Q)\mathrm{d}\x+\omega\int_{\partial \Omega}\Vert \Q-\Q_{bc}\Vert^2\mathrm{d}S \\
 &\geqslant \min\left(\frac{K}{2},C_1(A,B,C)\right)\Vert \Q\Vert_{W^{1,2}_{\Omega,\S_0}}^2-C_2(A,B,C),
\end{aligned}
\end{equation}
which means \eqref{energy} is coercive with respect to $\Q$. Now we prove the coerciveness estimate in $u$, i.e. if  $E(\Q,u)$ is bounded, then $u$ is also bounded in $W^{2,2}_\Omega$. The bulk energy $f_{bs}(u)$ is a fourth order polynomial of $u$ with $c>0$, and $\int_\Omega f_{bs}(u) \  \mathrm{d}\x$ is bounded, so $\Vert u \Vert_{L^2_\Omega}$, $\Vert u^2 \Vert_{L^2_\Omega}$ are also bounded. Similarly, $\Vert \Q \Vert_{L^2_{\Omega,\S_0}}$, $\Vert \Q^2 \Vert_{L^2_{\Omega,\S_0}}$ are shown to be bounded. 

When $A\geqslant \frac{B^2}{24C}$, the boundedness of $\Vert D^2 u\Vert^2_{L^2_\Omega}$ can be directly obtained from \eqref{eq: coupling term}. For $A<\frac{B^2}{24C}$, $\Vert D^2 u\Vert^2_{L^2_\Omega}$ is bounded by 
$$
\int_\Omega \left|D^2u\right|^2 \mathrm{d}\x\leqslant \int_\Omega 2\left|D^2u+q^2\left(\frac{\Q}{s_+}+\frac{\mathbf{I}_3}{3}\right)u\right|^2\mathrm{d}\x  +2\int_\Omega \left|q^2\left(\frac{\Q}{s_+}+\frac{\mathbf{I}_3}{3}\right)u\right|^2 \mathrm{d}\x.
$$
Given the boundedness of both $\Vert u \Vert^2_{L^2_\Omega}$ and $\Vert D^2 u\Vert_{L^2_\Omega}$ along with the inequality,
$
  \Vert u \Vert^2_{L^2_\Omega} +\Vert D^2 u\Vert_{L^2_\Omega}\geqslant C_3(\Omega) \Vert \nabla u \Vert^2_{L^2_\Omega}
$
from Theorem 5.19 of \cite{bedford2014calculus}, we have established the boundedness of $\Vert u \Vert_{W^{2,2}_\Omega}$ which proves the coerciveness estimate for $u$. The weak lower semi-continuity of the LdG energy and the surface energy is guaranteed in \cite{shi2024multistability} and the weak lower semi-continuity features of $f_{bs}$ and N-S coupling term are guaranteed in \cite{xia2023variational}. The existence of a global minimiser follows from the direct methods in the calculus of variations.
\end{proof}

For $A<\frac{B^2}{24C}$, the Euler-Lagrange equations of the free energy \eqref{energy} are given by 
\begin{equation}
\begin{aligned}
K\Delta \Q =& A \Q-B \left(\Q^2-\frac{tr(\Q^2)}{3}\mathbf{I}_3\right)+C tr(\Q^2)\Q\\
&+2B_0 q^2/s_+\cdot\left(u\cdot D^2u-\frac{tr(u\cdot D^2u)}{3}\mathbf{I}_3\right)+2 B_0 q^4\cdot \frac{\Q}{s_+^2}u^2,\\
2B_0\Delta^2 u =&- au-bu^2-cu^3-2B_0 D^2u:\left(q^2\left(\frac{\Q}{s_+}+\frac{\mathbf{I}_3}{3}\right)\right)\\
&-2B_0 \nabla\cdot \left(\nabla \cdot\left(q^2\left(\frac{\Q}{s_+}+\frac{\mathbf{I}_3}{3}\right)u\right)\right)-2B_0 \cdot \left|q^2\left(\frac{\Q}{s_+}+\frac{\mathbf{I}_3}{3}\right)\right|^2 u.
\end{aligned}
\label{E-L equation}
\end{equation}
where $\Delta^2u=\left(\frac{\partial^2}{\partial x^2_1}+\frac{\partial^2}{\partial x^2_2}+\frac{\partial^2}{\partial x^2_3}\right)^2 u$, and we prove that the weak solutions of \eqref{E-L equation}, $\bar{\Q} \in W_{\Omega,\S_0}^{1,2}$,$\bar{u}\in W_{\Omega}^{2,2}$, are in fact, classical solutions of \eqref{E-L equation}.

\begin{proposition}\label{regularity}
 Let $\Omega$ be a bounded, connected open set in $\mathbb{R}^3$, $\partial \Omega$ is $C^{4,1/2}$ continuous, and $K,B_0\neq 0$, then the weak solutions $\bar{\Q} \in W_{\Omega,\S_0}^{1,2},\bar{u}\in W_{\Omega}^{2,2}$ of \eqref{E-L equation} are classical solutions of \eqref{E-L equation}, i.e. $\bar{\Q} \in C_{\Omega,\S_0}^2$ and $\bar{u} \in C_\Omega^4$. 
\end{proposition}

\begin{proof}
Assume that $\bar{\Q} \in W_{\Omega,\S_0}^{1,2},\bar{u}\in W_{\Omega}^{2,2}$ are weak solutions of the following Euler-Lagrange equation,
\begin{equation}
\begin{aligned}
K\Delta \bar{\Q} =& \underbrace{A\bar{\Q}-B\left(\bar{\Q}^2-\frac{tr(\bar{\Q}^2)}{3}\mathbf{I}_3\right)+C tr(\bar{\Q}^2)\bar{\Q}}_{f_1(\bar{\Q})}\\
&+\underbrace{2 B_0  q^2/s_+\cdot \left(\bar{u}\cdot D^2\bar{u}-\frac{tr(\bar{u}\cdot D^2\bar{u})}{3}\mathbf{I}_3\right)}_{f_2(\bar{u})}+\underbrace{2 B_0 q^4 \cdot \frac{\bar{\Q}}{s_+^2}  \bar{u}^2}_{f_3(\bar{\Q},\bar{u})},\\
\Delta^2 \bar{u} =& \underbrace{- \frac{a}{2B_0}\bar{u}-\frac{b}{2B_0}\bar{u}^2-\frac{c}{2B_0}\bar{u}^3}_{f_4(\bar{u})}-\underbrace{D^2\bar{u}:\left(q^2\left(\frac{\bar{\Q}}{s_+}+\frac{\mathbf{I}_3}{3}\right)\right)}_{f_5(\bar{\Q},\bar{u})}\\
&-\underbrace{\nabla\cdot \left(\nabla \cdot\left(q^2\left(\frac{\bar{\Q}}{s_+}+\frac{\mathbf{I}_3}{3}\right)\bar{u}\right)\right)}_{f_6(\bar{\Q},\bar{u})}-\underbrace{\left|q^2\left(\frac{\bar{\Q}}{s_+}+\frac{\mathbf{I}_3}{3}\right)\right|^2 \bar{u}}_{f_7(\bar{\Q},\bar{u})}.
\end{aligned}
\label{E-L}
\end{equation}
From the density of $C^\infty_\Omega$ in $W^{1,2}_\Omega$ and $W^{2,2}_\Omega$ \cite{adams2003sobolev}, we can assume that the boundary data (or trace) of $\bar{u}$ and $\bar{Q}$ coincide with functions in $C^\infty_\Omega$. 

Recall that we are working in 3D case. By using the Sobolev embedding theorem in the 3D case \cite{adams2003sobolev}, we have
$
u\in W_{\Omega}^{2,2}\hookrightarrow C^{0,\frac{1}{2}}_\Omega, \Q\in W_{\Omega,\S_0}^{1,2}\hookrightarrow L^6_{\Omega,\S_0},
$
and then
$
f_1(\bar{\Q}), f_2(\bar{u}), f_3(\bar{\Q},\bar{u})\in L^2_{\Omega,\S_0}.
$
The right-hand side of the first partial differential equation is in $L^2_{\Omega,\S_0}$, and elliptic regularity yields $\Q \in W_{\Omega,\S_0}^{2,2}$, which is allowed by the regularity of boundary data and that of the domain \cite{evans2022partial}. Hence, we have
$
f_4(\bar{u})\in C^{0,\frac{1}{2}}_\Omega \subset L^2_\Omega, f_5(\bar{\Q},\bar{u}), f_6(\bar{\Q},\bar{u}), f_7(\bar{\Q},\bar{u})\in L^2_\Omega.
$
Then the right-hand side of the second partial differential equation in \eqref{E-L} is in $L^2_{\Omega,\S_0}$, and elliptic regularity yields $u \in W_{\Omega}^{4,2}$. Then, the right-hand side of the first equation of \eqref{E-L} belongs to $W^{2,2}_{\Omega,\S_0} \hookrightarrow C^{0,1/2}_{\Omega,\S_0}$, and the Schauder estimate \cite{landis1997second} gives $\Q \in C^{2,1/2}_{\Omega,\S_0}$. One can continue to alternately increase the regularity of $\bar{\Q}$ and $\bar{u}$ to obtain the full regularity.
\end{proof}


In the subsequent discussion, we will focus on the 2D (two-dimensional) case to facilitate comparisons with the experimental observations of smectic phases on square domains  \cite{cortes2016colloidal} and with the numerical results for nematic phases on 2D domains  \cite{yin2020construction,canevari2017order}. The results in Sections \ref{phase transition}, \ref{sec:far}, \ref{sec:near}, can be generalized to 3D cases, by employing the same methodology. In \cite{golovaty2015dimension}, the authors prove that in the thin film limit or for approximately 2D scenarios, and for certain choices of the surface energies that enforce tangent boundary conditions, the LdG energy minimisers are $z$-invariant, have a fixed eigenvector,  $\Vec{z}$ (the unit-vector in the $z$-direction) with associated fixed negative eigenvalue. This automatically reduces the number of degrees of freedom from $5$ to $2$, see below: 
\begin{equation}
\Q =
\left(
\begin{tabular}{cc}
	$\Q_{2D}+\frac{q_3}{6}\I_2$  & 0 \\
	 0  & -$\frac{q_3}{3} $
\end{tabular}
\right), \Q_{2D}=\left(
\begin{tabular}{cc}
	$q_1$  & $q_2$ \\
	 $q_2$  &  $-q_1$
\end{tabular}
\right),
\label{relationship between 2D and 3D tensor}
\end{equation}
where the constant, $q_3$, depends on the phenomenological parameters in the LdG energy and the anchoring coefficients in the surface energies. The symmetric, traceless $2\times 2$ matrix, $\Q_{2D}$ is often referred to as the rLdG order parameter \cite{han2020reduced}. 
Consequently, the LdG energy reduces to 
\begin{equation}
\begin{aligned}
E_{2D}(\Q_{2D},u)=\int_{\Omega_{2D}} &f_{bs}(u)+f_{int,2D}(\Q_{2D},u)\\
&+\frac{K}{2}\left| \nabla \Q_{2D} \right|^2+\underbrace{\frac{A_{2D}}{2}\text{tr}(\Q_{2D}^2)+\frac{C}{4}(\text{tr}(\Q_{2D}^2))^2}_{f_{bn,2D}(\Q_{2D})} \text{ } \mathrm{d}\x,
\end{aligned}
\label{energy_2D}
\end{equation}
where $\Omega_{2D} \subset \mathbb{R}^2$ (the 2D cross-section of the 3D domain in the thin-film limit), $A_{2D}=A-\frac{q_3B}{3}+\frac{q_3^2C}{6}$. When $A_{2D}<0$, the minimizer of $f_{bn,2D}$ is $\Q_{2D}=s_{+,2D}(\n_{2D} \times \n_{2D}-\mathbf{I}_2/2)$, where $\n_{2D}$ is an arbitrary 2D unit vector and
\begin{equation}
s_{+,2D}=\sqrt{-2A_{2D}/C},
\label{eq: 2D s}
\end{equation}
and thus, the 2D coupling energy density, $f_{int,2D}$, is defined to be:
\begin{equation}
f_{int,2D}(\Q_{2D},u)=
\begin{cases}
B_0 \left|D^2u+q^2\left(\frac{\Q_{2D}}{s_{+,2D}}+\frac{\mathbf{I}_2}{2}\right)u\right|^2, A_{2D}<0,\\
B_0\left|D^2u\right|^2, A_{2D}\geqslant 0.
\end{cases}
\end{equation}
For brevity, we omit the subscript, 2D, in $E_{2D}$, $f_{int,2D}$ $\Omega_{2D}$, $\Q_{2D}$, $s_{+,2D}$, $A_{2D}$. All subsequent results are based on the functional in \eqref{energy_2D}, also known as the rLdG energy in \cite{han2020reduced}.

\section{Thermotropic phase transition}\label{phase transition}
We consider the I-N-S phase transition with periodic boundary conditions. Consider the one-dimensional domain $\Omega=[0,h]$ and assume that the rLdG order parameter, $\Q$, is of the form
\begin{equation}
    \Q=\begin{pmatrix}
        Q & 0 \\ 
        0 & -Q 
    \end{pmatrix}.
\end{equation}
This corresponds to $\n_{2D} = (1,0)$ in the definition of $\Q_{2D}$ above so that there is only one degree of freedom: the scalar order parameter, $Q$, that measures the degree of ordering about the director. 
When $A<0$, the free energy \eqref{energy_2D} simplifies to
\begin{equation}
E_{1D}(Q,u)
=\int_0^h f_{bs}(u) + B_0\left[u_{xx}+q^2\left(\frac{Q}{\sqrt{-2A/C}} +\frac{1}{2}\right)u\right]^2+KQ_x^2+A Q^2+C Q^4 \text{ } \mathrm{d}x,
\label{energy:simple}
\end{equation}
and for $A\geqslant0$,
\begin{equation}
    E_{1D}(Q,u)
=\int_0^h f_{bs}(u) + B_0u_{xx}^2+KQ_x^2+A Q^2+C Q^4 \text{ } \mathrm{d}x.
\label{energy:simple A>0}
\end{equation}
The two temperature-dependent parameters are $A=\alpha_1(T-T_1^*)$, and $a=\alpha_2(T-T_2^*)$, where $T_2^*<T_1^*$. It is known that the isotropic phase loses stability for $T<T_1^*$, and we show that the nematic phase (with $u=0$) loses stability at $T<T_2^*$, and the smectic phase (with non-zero $u$) is the energy minimiser for $a<0$.  Hence, $T_1^*$ and $T_2^*$ are the critical temperatures for the I-N and N-S phase transitions respectively, with $T_2^* < T_1^*$ \cite{de1993physics}. 

The admissible spaces are 
\begin{equation}
\begin{cases}
    Q \in V_Q=\{\Q \in W_{\Omega}^{1,2}, Q(0)=Q(h), D_x Q(0)=D_x Q(h)\},\\
     u \in V_u=\{u\in W_{\Omega}^{2,2}, u(0)=u(h), D_x u(0)=D_x u(h), D_{xx} u(0)=D_{xx} u(h)\},
    \end{cases}
    \label{space for periodic u and Q}
\end{equation}
and the E-L equations 
for $A<0$ are
\begin{equation}
\begin{cases}
2KQ_{xx}=2AQ+4CQ^3+\frac{2B_0q^2uu_{xx}}{\sqrt{-2A/C}}+\frac{2B_0q^4}{\sqrt{-2A/C}}\left(\frac{Q}{\sqrt{-2A/C}}+\frac{1}{2}\right)u^2,\\
-2B_0u_{xxxx}=au+cu^3+4B_0q^2\left(\frac{Q}{\sqrt{-2A/C}}+\frac{1}{2}\right)u_{xx}+2B_0q^2\frac{Q_{xx}u}{\sqrt{-2A/C}}\\
\qquad \qquad \qquad+4B_0q^2\frac{Q_x u_x}{\sqrt{-2A/C}}+2B_0q^4\left(\frac{Q}{\sqrt{-2A/C}}+\frac{1}{2}\right)^2u.
\end{cases}
\label{E-L one D}
\end{equation}

\begin{proposition}\label{nematic loses stability}
   Let $c, B_0, K, C$ be positive constants, and let $q=\frac{2\pi n_0}{h}$ for a fixed positive integer $n_0$, where $n_0=1,2,3,\cdots$. As temperature decreases, the energy functional \eqref{energy:simple} exhibits second-order I-N phase transition at $T=T_1^*$, and nematic phase is stable for $T_2^*\leqslant T < T_1^*$, but loses stability when $T<T_2^*$.
\end{proposition}

\begin{proof}
The isotropic phase $(Q_{I}\equiv 0, u_{I}\equiv 0)$ is always a solution of the E-L equation of \eqref{energy:simple} for $A< 0$ and \eqref{energy:simple A>0} for $A\geqslant 0$, and the nematic phase $(Q_N\equiv \sqrt{-A/2C}, u_N\equiv 0)$ is the solution of \eqref{E-L one D} when $A< 0$. 

For $T\geqslant T_1^*$, we have $a=\alpha_2(T-T_2^*)\geqslant 0$, $A=\alpha_1(T-T_1^*)\geqslant 0$ i.e. $f_{bs}(u)\geqslant 0, AQ^2+CQ^4\geqslant 0$. Hence, for any $Q,u$ in admissible space, the isotropic phase $(Q_{I}\equiv 0, u_{I}\equiv 0)$ is the global minimiser for $T\geqslant T_1^*$, i.e.
\begin{equation}
    E_{1D}(Q,u)=\int_0^h f_{bs}(u) + B_0u_{xx}^2+KQ_x^2+A Q^2+C Q^4 \text{ } \mathrm{d}x \geqslant 0 = E_{1D}(Q_{I},u_{I}).
\end{equation}

For $T_1^*> T \geqslant T_2^*$, we have $a=\alpha_2(T-T_2^*)\geqslant 0$, i.e. $f_{bs}(u)\geqslant 0$. Hence, for any $Q,u$ in admissible space, nematic phase $(Q_N\equiv \sqrt{-A/2C}, u_N\equiv 0)$ is the global minimiser for $T_1^*> T \geqslant T_2^*$. 


To investigate the stability of nematic phase near $T=T^*_2$, we calculate the second variation of \eqref{energy:simple} at $(Q_N\equiv \sqrt{-A/2C}, u_N\equiv 0)$ for a periodic perturbation, $(\eta_1,\eta_2)$, 
\begin{equation}
\delta^2E_{1D}(\eta_1,\eta_2)=\int_0^h \left(a(T)\cdot\eta_2^2+2B_0\left(\eta_{2xx}+q^2\eta_2\right)^2 + 2K (\eta_{1x})^2 -4A \eta_1^2 \right)\mathrm{d}x.
\label{second variation}
\end{equation}
The stability of the nematic phase is measured by the minimum eigenvalue of $\delta^2E_{1D}$, 
\begin{equation}
\mu_T=\inf_{\eta_1\in V_Q,\eta_2 \in V_u}\frac{\delta^2E_{1D}(\eta_1,\eta_2)}{\int_0^h \eta_1^2 + \eta_2^2 \mathrm{d}x}.
\label{eigenvalue}
\end{equation}
If $\mu_T<0$, the nematic phase is unstable. If $\mu_T>0$, the nematic phase is stable.

For $T<T^*_1$, i.e., $-4A>0$, any perturbation with a non-zero $\eta_1$ is always a stable direction. Thus, we only consider the perturbation $(0,\eta_2)$. The Fourier expansion of the function, $\eta_2$, in $\Omega=[0,h]$ is given by
\begin{equation}
\eta_2=w_0/2+\sum_{n=1}^\infty w_n \cos(\frac{2\pi n x}{h}) + v_n \sin(\frac{2\pi n x}{h}).
\label{fourier expansion}
\end{equation}
By substituting \eqref{fourier expansion} into \eqref{second variation}, we have
$$
\delta^2 E_{1D}(0,\eta_2)=h/2\cdot\left(\frac{a+2B_0q^4}{2}\right)w_0^2+h/2\cdot\sum_{n=1}^\infty \left[2B_0\left(\frac{4\pi ^2 n^2}{h^2}-q^2\right)^2+a\right](w_n^2+v_n^2).
$$
$(0,\eta)$ is an eigenfunction of \eqref{second variation} if and only if 
\begin{equation}
    2a \eta + 4B_0 \eta_{xxxx}+8B_0q^2\eta_{xx}+4B_0q^4\eta=\lambda\eta.
    \label{KKT}
\end{equation}
One can verify that \eqref{KKT} is the first order optimal condition (or KKT condition \cite{boyd2004convex}) of \eqref{eigenvalue}. By substituting \eqref{fourier expansion} into \eqref{KKT}, we get that $\eta\equiv 1$, $\eta= \cos(\frac{2\pi n x}{h})$ and $\eta= \sin(\frac{2\pi n x}{h})$, $n=1,2,3\cdots$ are the eigenvectors of $\delta^2 E$ with eigenvalues $\mu=a+2B_0q^4$ and $a+2B_0\left(\frac{4\pi ^2 n^2}{h^2}-q^2\right)^2$, $n=1,2,3\cdots$, respectively. For $n_0\in\mathbb{Z^+}$ s.t. $\left(\frac{4\pi ^2 n_0^2}{h^2}-q^2\right)^2 = 0$,
then $\eta=\sin(\frac{2\pi n_0 x}{h}) = \sin\left( qx \right)$ and $\eta= \cos(\frac{2\pi n_0 x}{h}) = \cos\left( qx \right)$ are the eigenvectors corresponding to the minimum degenerate eigenvalue $\mu=a$.
For $T_1^* > T\geqslant T_2^*$, i.e., $a\geqslant0$, the second variation is always positive, i.e., the nematic phase is stable. For $T< T_2^*$, i.e., $a<0$, the eigenvector $\eta \equiv 1$ is an unstable eigendirection if and only if the corresponding eigenvalue $a+2B_0q^4<0$ is negative, and the eigenvectors $\sin(\frac{2\pi n x}{h})$ and $\cos(\frac{2\pi n x}{h})$, $n=1,2,3\cdots$, are unstable eigendirections if and only if the corresponding eigenvalue $a+2B_0\left(\frac{4\pi ^2 n^2}{h^2}-q^2\right)^2$ is negative. Thus, the Morse index of the nematic phase, i.e., the number of eigenvectors corresponding to negative eigenvalues is 
\begin{equation}\label{eq: Morse index of nematic}
i_{nematics} = 
    2\times card(\mathbb{N}_{nematics})+m_0,
\end{equation}
where
\begin{equation}\label{eq: index constraint}
\mathbb{N}_{nematics} = \left\{n\in\mathbb{Z^+}:a+2B_0\left(\frac{4\pi ^2 n^2}{h^2}-q^2\right)^2<0\right\}
\end{equation}
and $card(\mathbb{N}_{nematics})$ is the cardinal number of $\mathbb{N}_{nematics}$. If $a+2B_0q^4\geqslant 0$, $m_0=0$; otherwise $m_0=1$, i.e. $\eta \equiv 1$ is an unstable eigendiretion.
As $a<0$ 
decreases, more positive integers satisfy the constraint in \eqref{eq: index constraint}, and the Morse index of the nematic phase, $i_{nematics}$, increases. 

\end{proof}

\begin{remark}
The energy functional in \eqref{energy:simple} exhibits a second-order I-N phase transition, so that the isotropic and nematic phases cannot coexist in the 2D $\Q$-tensor model \eqref{energy_2D}. The 2D isotropic phase, \(\Q_{iso,2D} = 0\), is not equivalent to the 3D isotropic phase, \(\Q_{iso,3D} = 0\) (see \eqref{relationship between 2D and 3D tensor}). A first-order I-N phase transition can be demonstrated using a similar method that employs the full 3D $\Q$-tensor, with five degrees of freedom, in \eqref{energy}.
\end{remark}

For example, in Figure \ref{Morse index of nematic}, we substitute the parameter values in the caption, into \eqref{eq: Morse index of nematic}, and get $\mathbb{N}_{nematic} = \{3,4,5\}$ and $m_0 = 0$, i.e., the Morse index of the nematic phase $i_{nematics} = 6$ with unstable eigendirections $\eta=\sin(nx),\cos(nx)$, $n\in\mathbb{N}_{nematics}$.

\begin{figure}[ht]
\centering
\includegraphics[width=.7\textwidth]{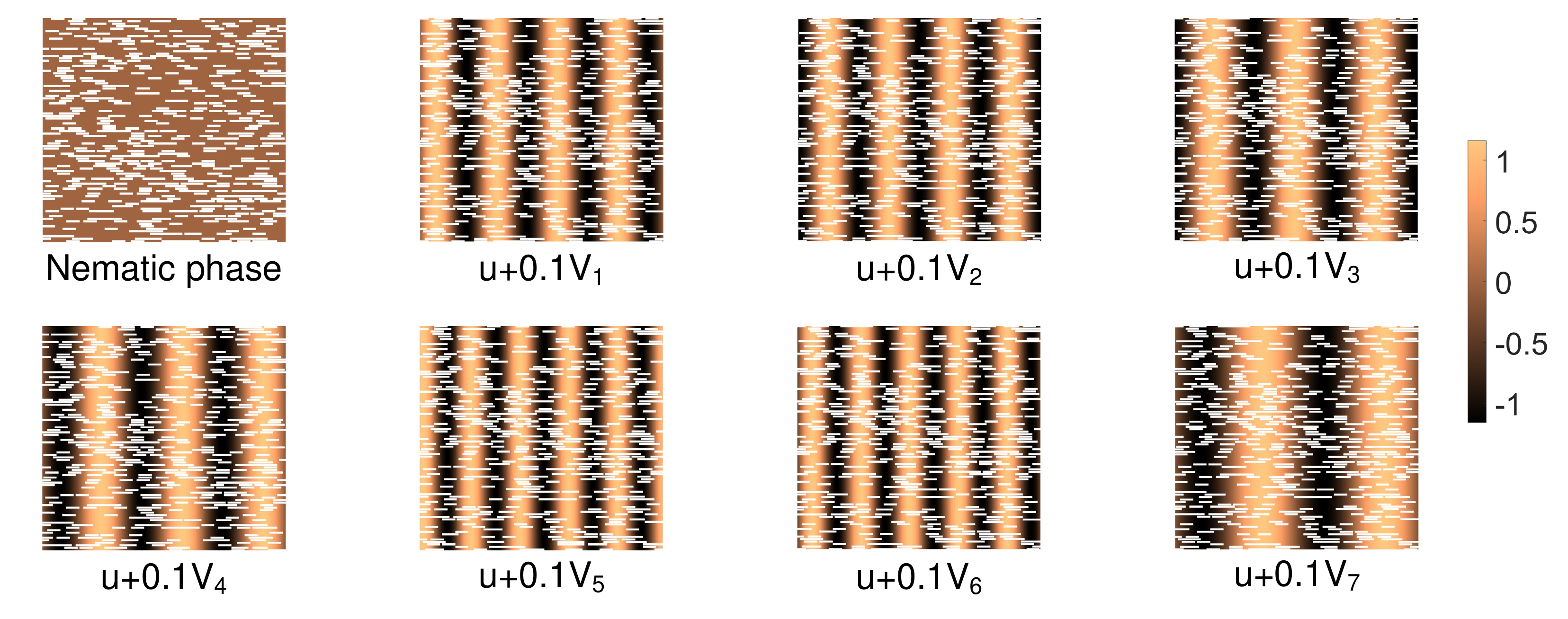}
\caption{The nematic critical point for $h=2\pi$, $q=4$,  $T=-30$, $T_1^*=0$, $T_2^*=-10$, $a=T-T_2^*=-20$, $A=T-T_1^*=-30$, $B_0=0.1$, $c=10, C=10$. $V_1$ to $V_6$ are the unstable eigendirections associated with $u$ and $V_7$ is a stable eigendirection. The pairs of unstable eigendirections $V_1$ and $V_2$, $V_3$ and $V_4$, $V_5$ and $V_6$ are the orthogonal linear combinations of $\sin(nx)$ and $\cos(nx)$ with $n =4,3,5$ respectively. The colour bar represents the modulation of the density, and we use the same visualization method in the following figures. The white lines define the nematic director.}
\label{Morse index of nematic}
\end{figure}

The aforementioned calculations show that the nematic phase loses stability as the temperature decreases. In the remainder of this section, we demonstrate that when the nematic phase loses stability, it bifurcates into a more stable smectic phase. To study this, we consider the E-L equation for $u$, 
\begin{equation}
    2B_0 u_{xxxx}+au+cu^3+4B_0q^2u_{xx}+2B_0q^4u=0,
    \label{E-L of u}
\end{equation}
i.e. fix $Q\equiv \frac{s_+}{2}$ in \eqref{E-L one D} for brevity, but the results also hold for variable $Q$. In the proof of Proposition \ref{nematic loses stability}, we note that the minimum eigenvalue of the nematic phase is degenerate, which presents technical difficulties in bifurcation theory 
\cite{chow2012methods}. To circumvent this issue, we construct the following working space:
\begin{equation}
    V=V_u \cap W^{1,2}_{0,\Omega},
\end{equation}
where $V_u$ is defined in \eqref{space for periodic u and Q}. This restricts $\eta=cos(qx)$ from serving as an eigenvector and then simplifies the minimum eigenvalue at the nematic phase.
\begin{proposition}\label{bifurcation proof}
    Given any positive $c, B_0$, and $q=\frac{2\pi n_0}{h}$, where $n_0$ is a natural number, a pitchfork bifurcation of \eqref{E-L of u} arises at $a=0$ or $T=T_2^*$, $u\equiv 0$ in $V$. More precisely,
there exists positive numbers $\epsilon,\delta$ and two smooth maps
\begin{equation}
    t\in (-\delta,\delta)\rightarrow a(t)\in(-\epsilon,\epsilon), t\in (-\delta,\delta)\rightarrow w_t \in V
\end{equation}
such that all the pairs $(a,u)\in R\times V$ satisfying 
$$
u \text{ } is \text{ } a \text{ } solution \text{ } to \text{ } \eqref{E-L of u}, |a|<\epsilon , \Vert u \Vert_{W^{2,2}_\Omega}\leqslant \epsilon
$$
are either
$$\text{ } nematic \text{ } phase: u\equiv0 \text{ } or\text{ } smectic \text{ } phases: u=\pm \left(t sin(qx)+ t^2 w_t\right).
$$
\end{proposition}
\begin{proof}
The proof follows the same paradigm as in Theorem 5.2 in \cite{lamy2014bifurcation} and Theorem 5.1 in \cite{canevari2017order}, and we address the necessary technical differences that arise because the study in \cite{lamy2014bifurcation} and  \cite{canevari2017order} focuses on a second-order partial differential equation, while our analysis involves a fourth-order partial differential equation.

To show that a pitchfork bifurcation arises at $a=0$, we apply the Crandall and
Rabinowitz bifurcation theorem \cite{crandall1971bifurcation} to the operator $\mathcal{F}:\mathbb{R}\times V \rightarrow W^{-2,2}_\Omega$ ($W^{-2,2}_\Omega$ is the dual space of $W^{2,2}_\Omega$) defined by
\begin{equation}
    \mathcal{F}(a,w):=2B_0D_{xxxx}w+aw+cw^3+4B_0q^2D_{xx}w+2B_0q^4w.
\end{equation}
We have to check four assumptions of Theorem 1.7 in \cite{crandall1971bifurcation}:

(a) $\mathcal{F}(a,0)=0$; (b) The partial derivatives $D_a\mathcal{F}, D_w\mathcal{F}, D_{aw}\mathcal{F}$ exist and are continuous; (c)$
    dim\left(\frac{W^{-2,2}(\Omega)}{Range(D_w\mathcal{F}(0,0))}\right)=dim\left(Kernel\left(D_w\mathcal{F}(0,0)\right)\right)=1$; (d) $ D_{aw}\mathcal{F}w_0\notin Range(D_{w}\mathcal{F}(0,0))$, where $w_0\in Kernel\left(D_w\mathcal{F}(0,0)\right)$. 

$\mathcal{F}(a,0)=0$ holds for all $a\in \mathbb{R}$. We have
\begin{equation}
\begin{cases}
    D_a\mathcal{F}(a,w)=w,\\
    D_w\mathcal{F}(a,w)=2B_0D_{xxxx}+a+3cw^2+4B_0q^2D_{xx}+2B_0q^4,\\
    D_{aw}\mathcal{F}(a,w)=1,
\end{cases}
\end{equation}
and they are continuous, since $D_w\mathcal{F}(a,w): V\rightarrow W^{-2,2}_\Omega$ is a bounded linear operator. For checking $\mathcal{F}$ satisfies assumption (c), we should calculate the kernel space of 
\begin{equation}
    D_w\mathcal{F}(0,0)=2B_0D_{xxxx}+4B_0q^2D_{xx}+2B_0q^4=2B_0(D_{xx}+q^2)(D_{xx}+q^2)
\end{equation}
in $V$, i.e. the solution space of the following differential equation:
\begin{equation}\label{eq: kernal space}
\begin{cases}
D_w\mathcal{F}(0,0)w=2B_0D_{xxxx}w+4B_0q^2D_{xx}w+2B_0q^4w=0,\\
w(0)=w(h)=0, D_{x}w(0)=D_{x}w(h), D_{xx}w(0)=D_{xx}w(h).
\end{cases}
\end{equation}
The general solution of the differential question in \eqref{eq: kernal space} without considering the boundary condition is 
\begin{equation}
    w=(k_1+k_2x)sin(qx)+(k_3+k_4x)cos(qx),k_i\in \mathbb{R},i=1,2,3,4.
\end{equation}
By taking the boundary condition into account, we have $w=k_1 sin(qx),k_1\in \mathbb{R}$, and
\begin{equation}
dim\left(Kernel\left(D_w\mathcal{F}(0,0)\right)\right)=dim\left(\{w=k_1 sin(qx),k_1\in \mathbb{R}\}\right)=1.
\end{equation}
For any $u_a,u_b\in V$, we have
\begin{equation}
\begin{aligned}
    \left<D_w\mathcal{F}(0,0) u_a,u_b\right>&=2B_0\int_0^h \left((D_{xx} +q^2)(D_{xx} +q^2)u_a\right)u_b \mathrm{d}x\\
    &= 2B_0\int_0^1 \left((D_{xx} +q^2)(D_{xx} +q^2)u_b\right)u_a \mathrm{d}x=\left< u_a,D_w\mathcal{F}(0,0)u_b\right>
    \end{aligned}
\end{equation}
by using the boundary conditions of $u_a$ and $u_b$, which means $D_w\mathcal{F}(0,0)$ is a self-adjoint operator, and hence it is a Fredholm operator of index $0$ \cite{ize1976bifurcation}. We have 
\begin{equation}
    dim\left(\frac{W^{-2,2}(\Omega)}{Range(D_w\mathcal{F}(0,0))}\right)=dim\left(Kernel\left(D_w\mathcal{F}(0,0)\right)\right)=1,
\end{equation}
which satisfies assumption (c). We also need to check the last assumption (d),
\begin{equation} 
    D_{aw}\mathcal{F}(a,w)sin(qx)=sin(qx)\notin range D_w\mathcal{F}(0,0),
\end{equation}
i.e. the following differential equation,
\begin{equation}
\begin{cases}
2B_0D_{xxxx}w+4B_0q^2D_{xx}w+2B_0q^4w=sin(qx),\\
w(0)=w(h)=0, D_x w(0)=D_x w(h)=0, D_{xx}w(0)=D_{xx}w(h),
\end{cases}
\label{assumption d}
\end{equation}
does not have a solution. One can check that the general solution of \eqref{assumption d} without considering the boundary condition  is 
\begin{equation}
    w=-\frac{x^2sin(qx)}{16B_0q^2}+k_1  sin(qx),k_1\in \mathbb{R},
\end{equation}
and it cannot satisfy the boundary conditions with any $k_1\in \mathbb{R}$, so that $sin(qx)\notin Range \left(D_w\mathcal{F}(0,0)\right)$. All the assumptions of Crandall and Rabinowitz's
theorem are satisfied, and the proposition follows directly from \cite{crandall1971bifurcation}.
\end{proof}

\begin{proposition}
    Given positive $c, B_0, K, C$, and $q=\frac{2\pi n_0}{h}$, where $n_0$ is a natural number, in \eqref{E-L one D}, 
    the nematic phase ($Q\equiv s_+/2,u\equiv 0$) loses stability in $V_Q\times V$ at the critical temperature $T=T_2^*$, via a symmetric pitchfork bifurcation.
\end{proposition}

\begin{proof}
In Proposition \ref{bifurcation proof}, we fix $Q\equiv \frac{s_+}{2}$ in  \eqref{E-L one D} for brevity, but the results also hold for coupled system \eqref{E-L one D} without treating $Q$ to be a constant by defining $\mathcal{F}(a,w_1,w_2)=(\mathcal{F}_1(a,w_1,w_2),\mathcal{F}_2(a,w_1,w_2)):\mathbb{R}\times V_Q \times V \rightarrow W^{-1,2}_\Omega \times W^{-2,2}_\Omega$
    where
    \begin{equation}  
    \begin{cases}
     \begin{aligned}  \mathcal{F}_1(a,w_1,w_2):=&-2KD_{xx}w_1+2A(a)(s_+(a)/2+w_1)+4C(s_+(a)/2+w_1)^3\\
    &+\frac{2B_0q^2w_2D_{xx}w_2}{s_+(a)}+\frac{2B_0q^4\left(1+\frac{w_1}{s_+(a)}\right)w_2^2}{s_+(a)},
    \end{aligned}\\   
    \begin{aligned}   \mathcal{F}_2(a,w_1,w_2):=&2B_0D_{xxxx}w_2+aw_2+cw_2^3+4B_0q^2\left(1+\frac{w_1}{s_+(a)}\right)D_{xx}w_2\\
    &+2B_0q^2\frac{w_2D_{xx}w_1}{s_+(a)}+4B_0q^2\frac{ D_x w_1 D_x w_2}{s_+(a)}+2B_0q^4\left(1+\frac{w_1}{s_+(a)}\right)^2w_2,
    \end{aligned}
    \end{cases}
    \end{equation}
$A(a)=\alpha_1(\frac{a}{\alpha_2}+T_2^*-T_1^*)$ and $s_+(a)=\sqrt{-2A(a)/C}$. The proof follows the same paradigm as in Proposition \ref{bifurcation proof}. One can check that
\begin{equation}
D_{(w_1,w_2)}\mathcal{F}(0,0,0)=\left(-2KD_{xx}-4A(0),2B_0(D_{xx}+q^2)(D_{xx}+q^2)\right)
\end{equation}
is also a Fredholm operator of index 0, and $dim(Kernel(D_{(w_1,w_2)}\mathcal{F}(0,0,0))=1$ since the spectrum \cite{ize1976bifurcation} of $-2KD_{xx}-4A(0),A(0)<0$ in $V_Q$ is positive which does not change the dimension of kernel space. 
\end{proof}

In Figure \ref{bifurcation}, we numerically calculate the N-S bifurcation, accomplished using the sine spectral method for $u$ and Fourier spectral method for $Q$ \cite{shen2011spectral} (see Appendix). This numerical scheme covers the boundary conditions in 
$V_Q\times V$. For $a>0$, the minimum eigenvalue at the nematic phase, as calculated both numerically and analytically, is both simple and positive, indicating stability. When $a=0$, a simple zero eigenvalue emerges with eigenvector $\eta=\sin(qx)$. As $a$ becomes negative, the nematic phase loses  stability and bifurcates into two smectic phases, corresponding to $u=t\sin(qx)+t^2 w_t$ and $u=-t\sin(qx)-t^2 w_t$ respectively, in pitchfork bifurcation.

The numerically calculated bifurcation diagram of the I-N-S phase transition v.s. temperature $T$ is shown in Figure \ref{Phase transition as temperature}. The isotropic phase with $u_I\equiv 0$ and $\Q_I\equiv \mathbf{0}$ is always a critical solution of \eqref{energy:simple}. When $T\geqslant T_1^*$, the isotropic phase is a global minimiser of \eqref{energy:simple}. For $T_1^* > T \geqslant T_2^*$, the isotropic phase loses stability, and the nematic phase with $u_N\equiv 0$ and $\Q_N\not\equiv \mathbf{0}$ becomes stable. For $T< T_2^*$, the nematic phase loses stability and the smectic phase with $u_S\not\equiv 0$ and $\Q_S\not\equiv \mathbf{0}$ becomes stable.

\begin{figure}[ht]
\centering
\includegraphics[width=.8\textwidth]{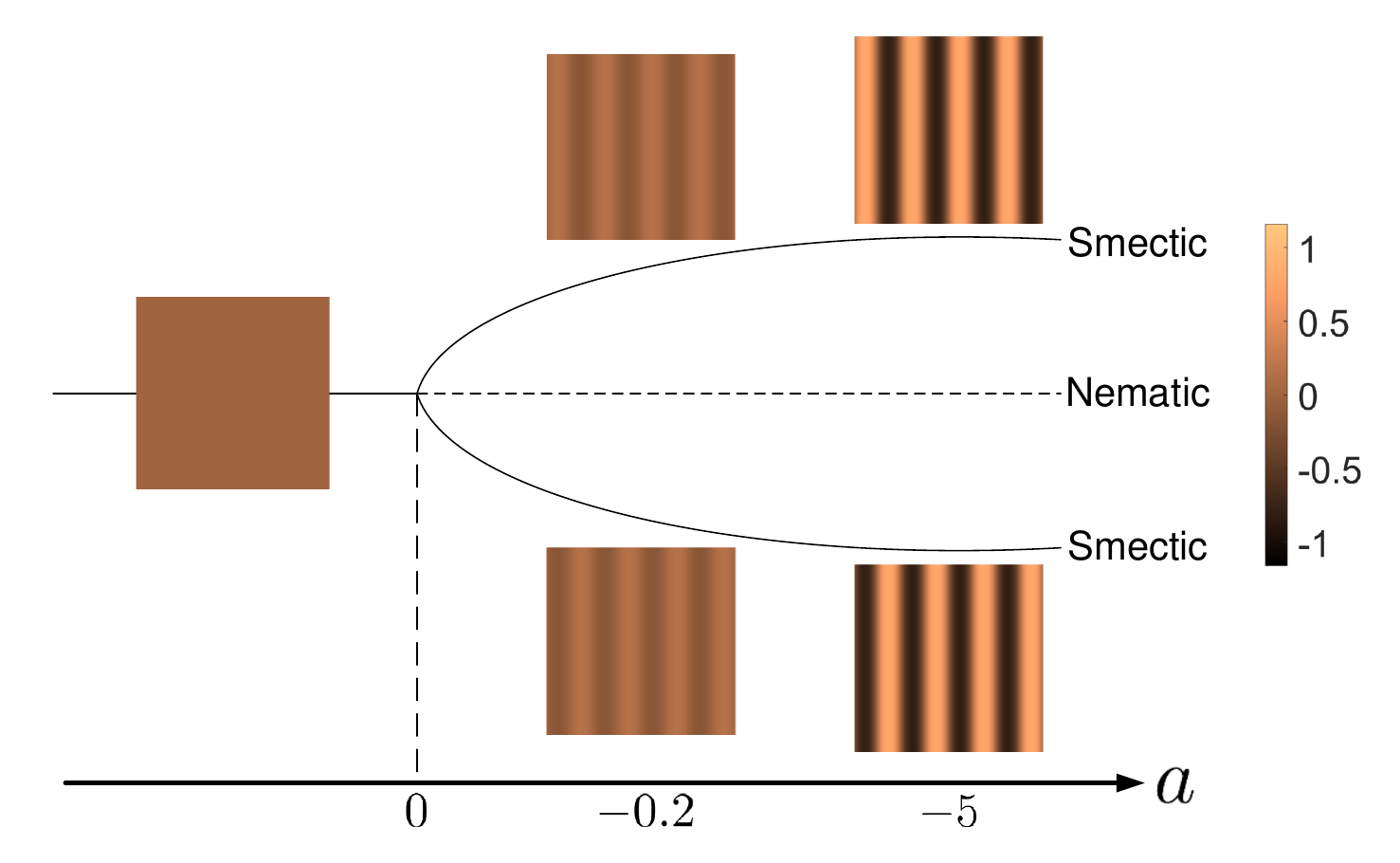}
\caption{Schematic illustration of the 
the N-S phase transition with $a=T+10$, $b=0$, $c=10$, $A=T$, $C=10$, $K=0.2$, $h=2\pi$, $q=4$, $B_0=0.001$, and the pitchfork bifurcation for $a<0$. The solid black line denotes a stable phase, while the dashed black line denotes an unstable phase in all figures. We numerically calculate the minimiser ($u_{min}$,$Q_{min}$) of \eqref{energy:simple} with various $a$, and plot $u_{min}$. 
We track the bifurcation 
across $-5\geqslant T \geqslant -15$ ($5\geqslant a \geqslant -5$).}
\label{bifurcation}
\end{figure}

\begin{figure}[ht]
\centering
\includegraphics[width=.8\textwidth]{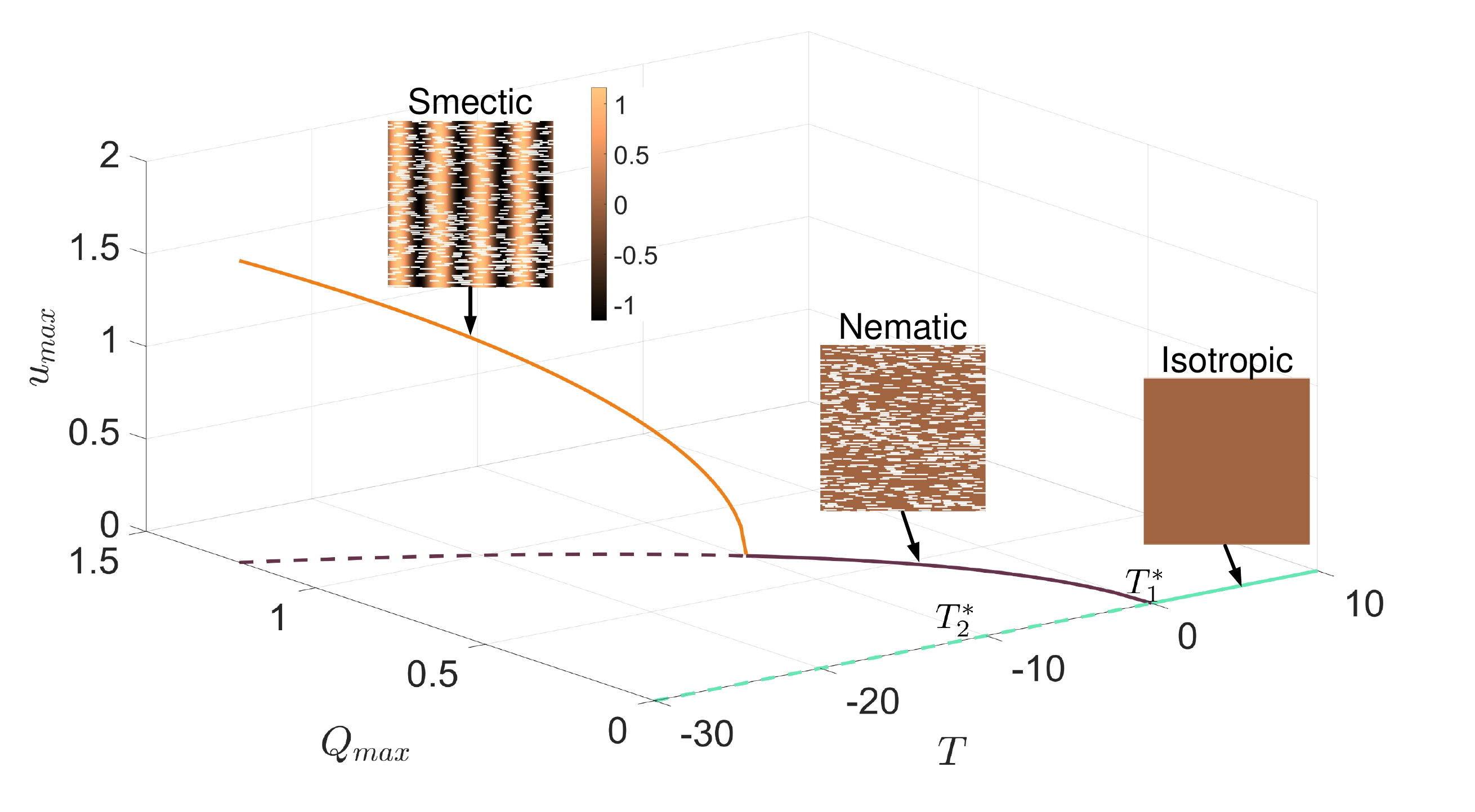}
\caption{Phase transitions for $T_1^*=0$, $T_2^*=-10$, $\alpha_1=\alpha_2=1$, $C=c=10$, $h=2\pi$, $q=4$, $B_0=0.001$. We use $u_{max}(T)$ and $Q_{max}(T)$, where $u_{max}(T)$ = $\max_{0\leqslant x \leqslant h} u^*_T(x)$ and $Q_{max}(T)$ = $\max_{0\leqslant x \leqslant h} Q^*_T(x)$, to represent the global minimizer $(Q^*_T(x),u^*_T(x))$ of $E_{1D}$ at $T$. For better visualisation, we plot the 2D $y$-invariants: $\bar{Q}(x,y)\equiv Q(x)$ and $\bar{u}(x,y) \equiv u(x)$.}
\label{Phase transition as temperature}
\end{figure}

\section{Smectics under confinement} \label{Sec: confinement}
In this section, we focus on the low temperature regime (i.e., $a < 0$ and $A < 0$) to investigate smectic profiles under confinement. In Sections \ref{sec:far} and \ref{sec:near}, we study the minimisers of the coupling energy, assuming a given rLdG $\Q$-profile, compatible with a defect-free perfectly ordered nematic state and a nematic defect respectively. These formal calculations give us some heuristic insight into how smectic layers respond to nematic profiles, with and without defects i.e. do defects repel smectic layers and do smectic layers concentrate near well-ordered nematic regions and if so, is there a correlation between the layer normal and the nematic director?
\subsection{The positional order far from defects}\label{sec:far}
Based on previous work \cite{liang2011nematic,yao2022defect}, we assume that far away from defects in confined geometries,
\begin{equation}\label{eq:Q}
\Q=s_+\left(\mathbf{n}\otimes \mathbf{n}-\frac{\mathbf{I}_2}{2}\right).
\end{equation} This models a perfectly ordered nematic state, which is also a minimiser of $f_{bn}$ in \eqref{energy_2D}, with arbitrary 2D nematic director $\mathbf{n}$.
Based on the analysis in Section \ref{phase transition}, we assume a simple periodic structure for $u$, compatible with a layer structure,, 
\begin{equation}\label{eq:u}
u(\mathbf{x})=k_1\cos(\tilde{q}\mathbf{k}\cdot \mathbf{x}),
\end{equation} 
where $\mathbf{k} = \frac{\nabla u}{|\nabla u|}$, if $|\nabla u|\neq 0$, is the layer normal, and $\tilde{q}$ is the wave number of the layer. Substituting \eqref{eq:Q} and \eqref{eq:u} into the coupling term \eqref{eq: coupling term}, we obtain
\begin{equation}
    \left|D^2u+q^2\left(\frac{\Q}{s_+}+\frac{\mathbf{I}_2}{2}\right)u\right|^2=k_1^2\left|-\tilde{q}^2\mathbf{k}\otimes \mathbf{k}\cdot u+q^2\mathbf{n}\otimes \mathbf{n}\cdot u\right|^2.
    \label{couple term}
\end{equation}
The above coupling term is minimised by $\tilde{q}=q$, $\mathbf{k}=\mathbf{n}$. Thus, we deduce that away from defects, we can interpret the phenomenological parameter $q$ in \eqref{energy_2D} to be the wave number of the smectic layers and the smectic layer normal is aligned with the nematic director, $\n$, in perfect agreement with the definition of SmA. Of course, these deductions do not shed light into the structure of arbitrary critical points of \eqref{energy_2D}. 

\subsection{The positional order near defects}\label{sec:near}
We can assume $\Q\equiv 0$ near defects in the rLdG model \cite{han2020reduced}. Substituting $\Q=0$ into the coupling term \eqref{eq: coupling term}, we obtain
\begin{equation}
    E_{couple}(\Q\equiv 0,u)=\int_{\Omega} B_0\left|D^2u+\frac{q^2u}{2}\mathbf{I}_2\right|^2 \mathrm{d}\x.
    \label{energy_couple}
\end{equation}
It's straightforward to verify that $u\equiv 0$ is a global minimiser since $E_{couple}(\Q\equiv 0,u)\geqslant 0 =E_{couple}(\Q\equiv 0,u\equiv 0)$. Our aim is to demonstrate that $u\equiv 0$ is indeed the unique minimiser, which implies that domains with defects do not support layered structures. We prove (a) $E_{couple}(\Q\equiv0,u)$ is convex, so that every minimiser $u^*$ is a global minimiser, i.e. $E_{couple}(u^*)=0$, and (b) if $E_{couple}(u^*)=0$, then $u^*\equiv 0$. (a) is obvious, since \eqref{energy_couple} is a squared norm of $\left(D^2u+\frac{q^2u}{2}\mathbf{I}_2\right)$. Next, we prove (b). If $\left|D^2u+\frac{q^2u}{2}\mathbf{I}_2\right|^2\equiv0$, then $u_{xy}\equiv0$, $u_{xx}\equiv u_{yy} =-\frac{q^2u}{2}$. From the regularity result in Proposition \ref{regularity}, we can assume that $u$ has $C^3$ regularity. Since $u_{xy} \equiv 0$, then $u_{xxy}=-\frac{q^2u_y}{2}\equiv0$, $u_{xyy}=-\frac{q^2u_x}{2}\equiv0$, which imply $u_x=u_y\equiv0$, and further $u\equiv C_0$ where $C_0$ is a constant. Then we deduce $C_0=0$ from $u_{xx}= u_{yy}=-\frac{q^2u}{2}\equiv 0$. Hence, (a) and (b) hold, which means that $u\equiv 0$ is the unique minimiser of \eqref{energy_couple}.

In Figure \ref{distribution_near_defect}, given a $\Q$-field on a square domain with edge length $\lambda$ and natural boundary conditions for $\Q$ and $u$, we plot the numerical minimiser, $u$, of \eqref{energy_2D} with relatively large $B_0$ and relatively small $a$ and $c$. The $u$ almost vanishes at the central point defect and produces a layered structure far away from the defect, in agreement with our analysis above. 

\begin{figure}[ht]
\centering
\includegraphics[width=.5\textwidth]{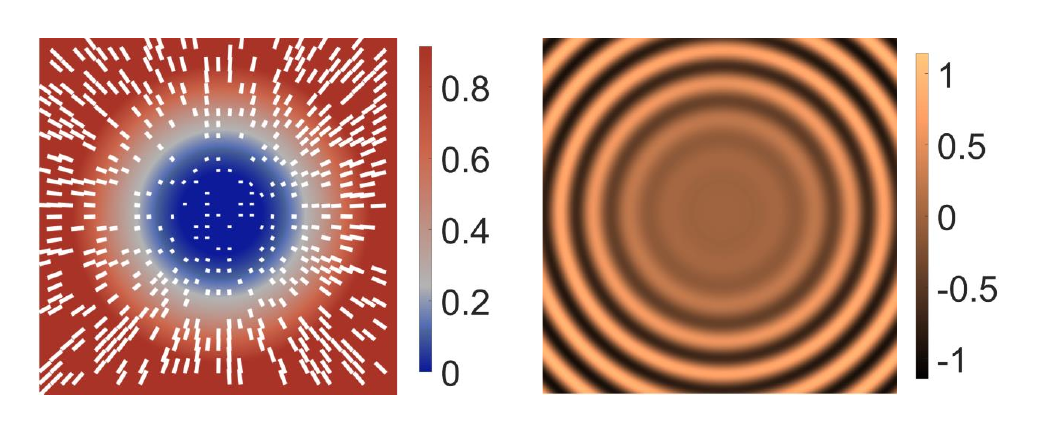}
\caption{The energy-minimising profile for $u$ with a fixed $\Q$-field  on the left. This $\Q$-field has a $+1$ central point defect. We use $a=-0.1$, $c=0.1$, $\lambda^2=30$, $q=2\pi$, $B_0=10^{-3}$. The colour bar of left plot is the order parameter $\sqrt{\text{Tr}(\Q^2)/2}$ and the white lines define the nematic director. We use the same color bar for $u$ as before.} 
\label{distribution_near_defect}
\end{figure}

\subsection{Structural transitions for smectics on square domains}\label{Confined Smectic}
We consider qualitative properties of energy minimisers of \eqref{energy_2D} on 2D square domains: $\Omega=[-\lambda,\lambda]^2$. By rescaling the system according to $\bar{\x} = \frac{\x}{\lambda}$, $\bar{E}=\frac{E
}{K}$, $\bar{\lambda}^2=\frac{2C\lambda^2}{K}$, $\bar{a}=\frac{a}{2C}$, $\bar{c}=\frac{c}{2C}$, $\bar{q}^2=\frac{Kq^2}{2C}$, $\bar{B}_0=\frac{2B_0C}{K^2}$, $\bar{A}=\frac{A}{2C}$ where the unit of $B_0$ is Nm$^2$, the unit of $K$ is N, the unit of $\lambda$ is m, and the unit of $q$ is m$^{-1}$. Then the non-dimensionalised energy is given by
\begin{equation}
\begin{aligned}
    \bar{E}(\Q,u)=\int_{[-1,1]^2} & \left( \bar{\lambda}^2 \left(\frac{\bar{a}}{2}u^2+\frac{\bar{c}}{4}u^4\right)+\frac{\bar{B}_0}{\bar{\lambda}^2}\left|D^2u+\bar{\lambda}^2\bar{q}^2\left(\frac{\Q}{s_+}+\frac{\mathbf{I}_2}{2}\right)u\right|^2 \right.\\
    &\quad+\left.\frac{1}{2}\left|\nabla \Q\right|^2+\bar{\lambda}^2\left(\frac{\bar{A}}{2}\text{tr}\Q^2+\frac{(\text{tr}\Q^2)^2}{8}\right) \right)\mathrm{d}\bar{x}.
    \label{non-dimensionalised energy}
\end{aligned}
\end{equation}

In the following, we drop all the bars, and the E-L equations of \eqref{non-dimensionalised energy} are
\begin{equation}
\begin{aligned}
\Delta \Q =& \lambda^2\left(A\Q+\frac{\text{tr}(\Q^2)\Q}{2}\right)\\
&+2B_0 q^2/s_+\cdot \left(u\cdot D^2u-\frac{tr(u \cdot D^2u)}{2}\mathbf{I}_2\right)+2\lambda^2B_0 q^4 \frac{\Q}{s_+^2}u^2,\\
\Delta^2 u =&- \lambda^4\left(\frac{a}{2B_0}u+\frac{c}{2B_0}u^3\right)-\lambda^2 D^2u:\left(q^2\left(\frac{\Q}{s_+}+\frac{\mathbf{I}_2}{2}\right)\right)\\
&-\lambda^2 \nabla\cdot \left(\nabla \cdot\left(q^2\left(\frac{\Q}{s_+}+\frac{\mathbf{I}_2}{2}\right)u\right)\right)-\lambda^4\left|q^2\left(\frac{\Q}{s_+}+\frac{\mathbf{I}_2}{2}\right)\right|^2 u.
\end{aligned}
\label{non-dimensionlized E-L equation}
\end{equation}
Regarding the boundary conditions, we assume Dirichlet tangent boundary conditions for the nematic director i.e. the director, $\mathbf{n}=\pm(1, 0)$ on the horizontal edges and $\mathbf{n}=\pm(0, 1)$ on the vertical edges, and the density is naturally distributed, i.e.,
\begin{equation} \label{tangent}
\begin{cases}
\Q= \begin{pmatrix}
        s_+L(x)/2  & 0 \\
        0 & -s_+L(x)/2  
    \end{pmatrix} \text{ on } y=\{1,-1\}, \\
    \Q= \begin{pmatrix}
        -s_+L(y)/2 & 0 \\
        0 & s_+L(y)/2 
    \end{pmatrix} \text{ on } x=\{1,-1\},\\
\left(D^2u+\lambda^2 q^2\left(\frac{\mathbf{Q}}{s_+}+\frac{\I_2}{2}\right)u\right): \vec{\nu} \otimes \vec{\nu}=0, \\
\left[\nabla \cdot \left(D^2u+\lambda^2 q^2\left(\frac{\Q}{s_+}+\frac{\I_2}{2}\right)u\right) \right]\cdot\vec{\nu}+\nabla \left[\left(D^2u+ \lambda^2 q^2\left(\frac{\mathbf{Q}}{s_+}+\frac{\I_2}{2}\right)u\right): \vec{\nu}\otimes \vec{\tau} \right]\cdot \vec{\tau}=0, 
\end{cases}
\end{equation}
where $\vec{\tau}$ is the tangential vector,
\begin{equation}
L(x)=
\begin{cases}
        \frac{x+1}{\epsilon_0}, -1\leqslant x\leqslant -1+\epsilon_0,\\
        1,|x|\leqslant 1-\epsilon_0,\\
        \frac{1-x}{\epsilon_0}, 1-\epsilon_0\leqslant x\leqslant 1,
\end{cases}
\end{equation}
is a trapezoidal function with a small enough $\epsilon_0$, to avoid the mismatch in the boundary conditions, at the square vertices \cite{shi2022nematic,canevari2017order,robinson2017molecular}.

\subsubsection{Large domain size limit}
In the $\lambda \rightarrow \infty$ limit or in the Oseen-Frank limit, we can assume that the interior profile is almost a minimiser of $f_{bn}$ in \eqref{energy_2D} with no defects \cite{10.1093/imamat/hxad031}. In the Oseen-Frank limit, the interior profile is
\begin{equation}
    \Q \equiv s_+\left(\n_0\times\n_0-\frac{\I}{2}\right),
    \label{assumption Q}
\end{equation} where $\n_0 = (\cos \theta, \sin \theta)$ and $\theta$ is a solution of the Laplace equation, subject to Dirichlet conditions compatible with \eqref{tangent}; the condition on $\theta$ originates from the nematic elastic energy. However, numerical results show that $\n_0$ is often constant away from the square edges, particularly for large $\lambda$ \cite{luo2012multistability}. 
Analogous to the discussion in Section \ref{sec:far}, we assume a constant $\n_0$ or $\theta$ in \eqref{assumption Q}  and assume a periodic structure for
\begin{equation}
    u=A_0\cos(\mathbf{k}\cdot \mathbf{x})
    \label{assumption u}
\end{equation}
with unknown $A_0$ and $\mathbf{k}$, where $A_0$ is the amplitude of the layers, $\frac{|\mathbf{k}|}{2\pi}$ (if $|k|\neq 0$) is the wave number of layers, and $\frac{\mathbf{k}}{|\mathbf{k}|}$ is the layer normal. 


By substituting \eqref{assumption u} and \eqref{assumption Q}, we have that the leading order terms in \eqref{non-dimensionalised energy}, in the $\lambda \to \infty$ limit are: 
\begin{equation}
\begin{aligned}      &\lambda^2\int_{\Omega}\left(\frac{au^2}{2}+\frac{cu^4}{4}+\frac{B_0}{\lambda^4}\left|D^2 u+\lambda^2 q^2\left(\frac{\Q}{s_+}+\frac{\mathbf{I_2}}{2}\right)u\right|^2+\frac{A}{2}\text{tr}\Q^2+\frac{(\text{tr}\Q^2)^2}{8} \right)\mathrm{d}\x\\
        &=\lambda^2\left(aA_0^2+\frac{3cA_0^4}{8}+2B_0A_0^2\left|q^2\n_0\times\n_0-\frac{\mathbf{k}\times\mathbf{k}}{\lambda^2}\right|^2+Constant+O\left(\frac{1}{|\mathbf{k}|}\right)\right).
        \label{assumption constraint}
\end{aligned}
\end{equation}
The leading order energy in \eqref{assumption constraint} is minimised by
\begin{equation}
    \mathbf{k}=q \lambda \n_0, A_0=\sqrt{\frac{-4a}{3c}},
    \label{wave_amplitude}
\end{equation} since the constant can be set to zero by adding a suitable constant to $f_{bn}$ in \eqref{assumption constraint}. These relations contain useful information: (i) the layer normal is aligned with $\n_0$; (ii) the number of layers is proportional to $\lambda$ and the (iii) layer thickness, $l$ is inversely proportional to $q$, in the $\lambda \to \infty$ limit. Further, the amplitude of the layer oscillations, $A_0$, depends on the parameters of $f_{bs}$ as expected, at least for energy minimisers in the $\lambda \to \infty$ limit.
In the first and second pairs of plots in Figure \ref{figure 5}, we fix $\n_0=(\sqrt{2}/2,\sqrt{2}/2)$ in \eqref{assumption Q}, and numerically calculate the minimiser of $u$ in \eqref{assumption constraint}, without assuming the periodic profile of $u$ in \eqref{assumption u}. In the numerical results, the wave number is proportional to $\lambda$; the layer normal follows the director $\mathbf{n_0}$; the amplitude of $u$ is close to $A_0$ in \eqref{wave_amplitude}. More specifically,  the numerically computed number of layers for $\lambda^2=50$ is $20$, which is equal to the predicted value $\frac{|\mathbf{k}|*2\sqrt{2}}{2\pi}=20$ (where $\frac{|\mathbf{k}|}{2\pi}$ denotes the number of layers in a unit length, and $2\sqrt{2}$ is the square diagonal length) in \eqref{wave_amplitude}, and the amplitude is 1.1432, close to the predicted value $A_0=\sqrt{\frac{-4a}{3c}}\approx 1.1547$ in \eqref{wave_amplitude}. For $\lambda^2 = 150$, the number of layers is $35$ and the predicted value is $\frac{|\mathbf{k}|*2\sqrt{2}}{2\pi}=34.6410$ in \eqref{wave_amplitude}. The numerically calculated amplitude is 1.1403, while the predicted value is $A_0=\sqrt{\frac{-4a}{3c}}\approx 1.1547$ in \eqref{wave_amplitude}. In the third pair plotted in Figure \ref{figure 5}, the director field is compatible with the boundary condition \eqref{tangent}. The number of layers along the diagonal is also $35$, and the numerically calculated amplitude is $1.1474$, both of which are also close to the predicted values.

\begin{figure}[ht]
\centering
\includegraphics[width=.98\textwidth]{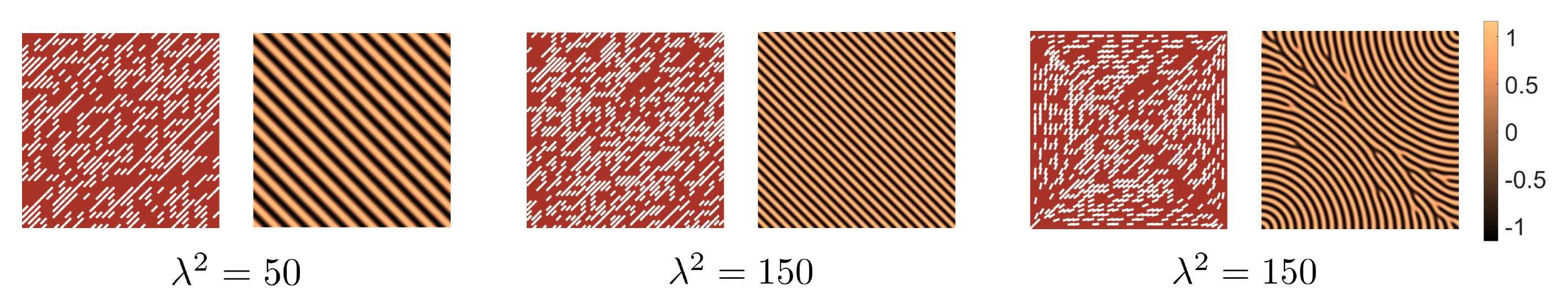}
\caption{The distribution of layers, u, is calculated by minimising \eqref{non-dimensionalised energy} with fixed $\Q$-field. In the left two plots, $\Q_{11}=0,\Q_{12}=\frac{s_+}{2}$, i.e. the nematic director is uniformly aligned along the line $y=x$. This is not compatible with the boundary conditions in \eqref{tangent}. In the right plot, $Q_{11}=s_+\cos(2\theta)/2$, $Q_{12}=s_+\sin(2\theta)/2$, $\theta$ is a solution of the Laplace equation, compatible with the boundary conditions in \eqref{tangent}. The parameters are $a=-5$, $c=5$, $B_0=10^{-3}$, $q=2\pi$, $A=-0.8359$.}
\label{figure 5}
\end{figure}
\subsubsection{Small domain size limit}
There is a unique global minimiser of the LdG energy on square domains, in the $\lambda \to 0$ limit, known as the Well Order Reconstruction Solution (WORS) \cite{yin2020construction,canevari2017order,kralj2014order} with two crossed line defects along the square diagonals. In this subsection, we show that in the $\lambda \rightarrow 0$ limit, the stable smectic critical points of \eqref{energy_2D} have the WORS as their $\Q$-profile, i.e. $\Q \rightarrow \Q_{WORS}$ and $u$ does not have a layer structure. In the $\lambda \rightarrow 0$ limit, we take a regular perturbation expansion of $\Q$ and $u$ in powers of $\lambda$ as shown below:
\begin{equation}
    \Q=\Q_0+\lambda\Q_1+\lambda^2\Q_2+\cdots,
    u=u_0+\lambda u_1+\lambda^2 u_2 + \cdots
\end{equation}
where $(\Q_0,u_0)$ is the solution of the following partial differential equation:
\begin{equation}
\begin{cases}
\Delta \Q_0 =2B_0\cdot q^2/s_+\cdot\left(u_0\cdot D^2u_0-\frac{tr(u_0\cdot D^2u_0)}{2}\mathbf{I_2}\right)\\
\Delta^2 u_0 =0
\end{cases}
\label{limit equation},
\end{equation}
subject to the boundary condition:
\begin{equation}
    \begin{cases}
\Q_0= \begin{pmatrix}
        s_+L(x)/2 & 0 \\
        0 & -s_+L(x)/2 
    \end{pmatrix} \text{ on } y=\{1,-1\}, \\
    \Q_0= \begin{pmatrix}
        -s_+L(y)/2 & 0 \\
        0 & s_+L(y)/2 
    \end{pmatrix} \text{ on } x=\{1,-1\}\\
    D^2u: \vec{\nu} \otimes \vec{\nu}=0, 
    \left[\nabla \cdot \left(D^2u\right) \right]\cdot\vec{\nu}+\nabla \left(D^2u: \vec{\nu}\otimes \vec{\tau}\right)\cdot \vec{\tau}=0, \text{ on } \partial \Omega.
\label{limit boundary condition}
\end{cases}
\end{equation}
\begin{proposition}\label{proposition:lambda limit to zero}
    The solutions of \eqref{limit equation} with boundary conditions \eqref{limit boundary condition} are 
    \begin{equation}
    \begin{cases}
        \Q_0(x,y)= \begin{pmatrix}
        Q_0(x,y) & 0 \\
        0 & -Q_0(x,y) 
    \end{pmatrix},\\
    u_0(x,y)=k_1x+k_2y+k_3,k_i \in \mathbb{R},i=1,2,3,
    \end{cases}
    \end{equation}
where
\begin{equation}\label{eq: WORS}
\begin{aligned}
    Q_0(x,y)=&\sum_{k \text{ odd}}\frac{4 s_+\sin\left(\frac{k\pi\epsilon_0}{2}\right)}{k^2\pi^2\epsilon_0} \sin\left(\frac{k\pi (x+1)}{2}\right)\frac{\sinh\left(\frac{k \pi \left(1-y\right)}{2}\right)+\sinh\left(\frac{k \pi \left(1+y\right)}{2}\right)}{sinh(k\pi)}\\
    &-\sum_{k \text{ odd}}\frac{4s_+ \sin\left(\frac{k\pi\epsilon_0}{2}\right)}{k^2\pi^2\epsilon_0} \sin\left(\frac{k\pi (y+1)}{2}\right)\frac{\sinh\left(\frac{k \pi \left(1-x\right)}{2}\right)+\sinh\left(\frac{k \pi \left(1+x\right)}{2}\right)}{\sinh(k\pi)}.
\end{aligned}
\end{equation}
\end{proposition}

\begin{proof}
Since the boundary-value problem for $u_0$ is not dependent on $\Q_0$, we note that $u_0$ is actually the critical point of the following energy functional,
\begin{equation} \label{eq:lambda0}
    E_0(u)=\int_{[-1,1]^2}|D^2 u|^2 \mathrm{d}\x,
\end{equation}
with natural boundary conditions. $E_0(u)$ is convex on $u$, and thus all the critical points are the global minimiser, i.e. $E_0(u_0)=0$. Consequently, $u_0$ satisfies $D_{xx}u_0=D_{yy}u_0=D_{xy}u_0\equiv 0$, so that $u_0$ is a linear function, 
\begin{equation}
    u_0=k_1x+k_2y+k_3,k_i \in \mathbb{R},i=1,2,3.
\end{equation}
Given a linear $u_0$, the partial differential equation of $\Q_0$ simplifies to 
\begin{equation}
\begin{cases}
    \Delta \Q_0 =0,\\
    \Q_0= \begin{pmatrix}
        s_+L(x)/2 & 0 \\
        0 & -s_+L(x)/2 
    \end{pmatrix} \text{ on } y=\{1,-1\},\\
    \Q_0= \begin{pmatrix}
        -s_+L(y)/2 & 0 \\
        0 & s_+L(y)/2 
    \end{pmatrix} \text{ on } x=\{1,-1\},
\end{cases}
\end{equation}
which can be solved by the separation of variables. A standard computation for the WORS profile as in \cite{fang2020surface} yields the results in \eqref{eq: WORS}.
\end{proof}
\begin{remark}
Clearly, the solution of \eqref{limit equation} with the boundary condition \eqref{limit boundary condition} is not unique because all linear functions, $u_0$ are compatible with the leading order problem \eqref{limit equation}. 
The implication supports our physical intuition that small domains cannot accommodate layer structures. One can directly check that the eigenvector of $\Q_0$ is either horizontal or vertical, and $\Q_0(x,x)=\Q_0(x,-x)=0$, which means $\Q_0$ has two line defects along the diagonals of square (also see Figure \ref{Asymptotic}).
\end{remark}

Next, we solve for $Q_1,Q_2,u_1,u_2$ to examine the effects of small perturbations, for small but non-zero $\lambda$. Up to $O(\lambda)$, the governing partial differential equations for $\Q_1$ and $u_1$ are
\begin{equation}
\begin{cases}
    \Delta^2 u_1=0,\\
    \Delta \Q_1=0
    \end{cases}
    \label{first order term equation}
\end{equation}
with the boundary conditions
\begin{equation}
\begin{cases}
    \Q_1=0, \text{ on } \partial \Omega,\\ 
    D^2u_1\cdot \vec{\nu}=\mathbf{0}, \left[\nabla \cdot D^2u_1\right] \cdot\vec{\nu}=0, \text{ on } \partial \Omega.
    \end{cases}
    \label{first order term boundary condition}
\end{equation}
which only has the trivial solution, i.e. $Q_1\equiv 0$ and linear $u_1$. Up to $O(\lambda^2)$, the governing partial differential equations for $\Q_2$ and $u_2$ are
\begin{equation}
\begin{cases}
    \Delta \Q_2= A\cdot\Q_0+\frac{\text{tr}(\Q_0^2)\Q_0}{2}\\
    \qquad  \qquad +2B_0\cdot q^2/s_+\cdot \left(u_0\cdot D^2u_2-\frac{tr(u_0 \cdot D^2u_2)}{2}\mathbf{I}_2\right)+2B_0\cdot q^4\cdot \frac{\Q_0}{s_+^2}u_0^2,\\
        \Delta^2 u_2=-\nabla\cdot\left(\nabla\cdot\left(q^2\left(\frac{\Q_0}{s_+}+\frac{\I_2}{2}\right)u_0\right)\right),
    \end{cases}
    \label{second order term equation}
\end{equation}
with the boundary conditions
\begin{equation}
\begin{cases}
    \Q_2=0,  \\ 
   \left(D^2u_2+ q^2\left(\frac{\Q_0}{s_+}+\frac{\I_2}{2}\right)u_0\right): \vec{\nu}\otimes \vec{\nu}=0,\\
    \left[\nabla \cdot \left(D^2u_2+ q^2 \left(\frac{\Q_0}{s_+}+\frac{\mathbf{I}_2}{2}\right)u_0\right) \right]\cdot\vec{\nu}+\nabla\left[\left(D^2u_2+ q^2\left(\frac{\mathbf{Q}_0}{s_+}+\frac{\I_2}{2}\right)u_0\right): \vec{\nu}\otimes \vec{\tau} \right]\cdot \vec{\tau}=0.
    \end{cases}
    \label{second order term boundary condition}
\end{equation}
The differential equation for $\Q_2$ can be numerically solved using the finite difference method, but the boundary-value problem for $u_2$ is difficult to solve because of the complex boundary condition, which involves the second and third derivatives. Fortunately, the solution of \eqref{second order term equation} with the boundary condition \eqref{second order term boundary condition} is a critical point of the following functional
\begin{equation}
    \Tilde{E}(u_2)=\int_{[-1,1]^2} \left|D^2u_2+q^2\left(\frac{\Q_0}{s_+}+\frac{\mathbf{I}_2}{2}\right)u_0\right|^2 \mathrm{d}\x,
\end{equation}
without any boundary anchoring. By minimizing the above energy, we can numerically calculate $u_2$, which exhibits some oscillation along the directors of WORS, as shown in Figure \ref{Asymptotic}(a). For $\lambda^2=0.01$, the density distribution, $u$, is no longer a linear function and tends to demonstrate a layer structure.
\begin{figure}
\centering
\includegraphics[width=.99\textwidth]{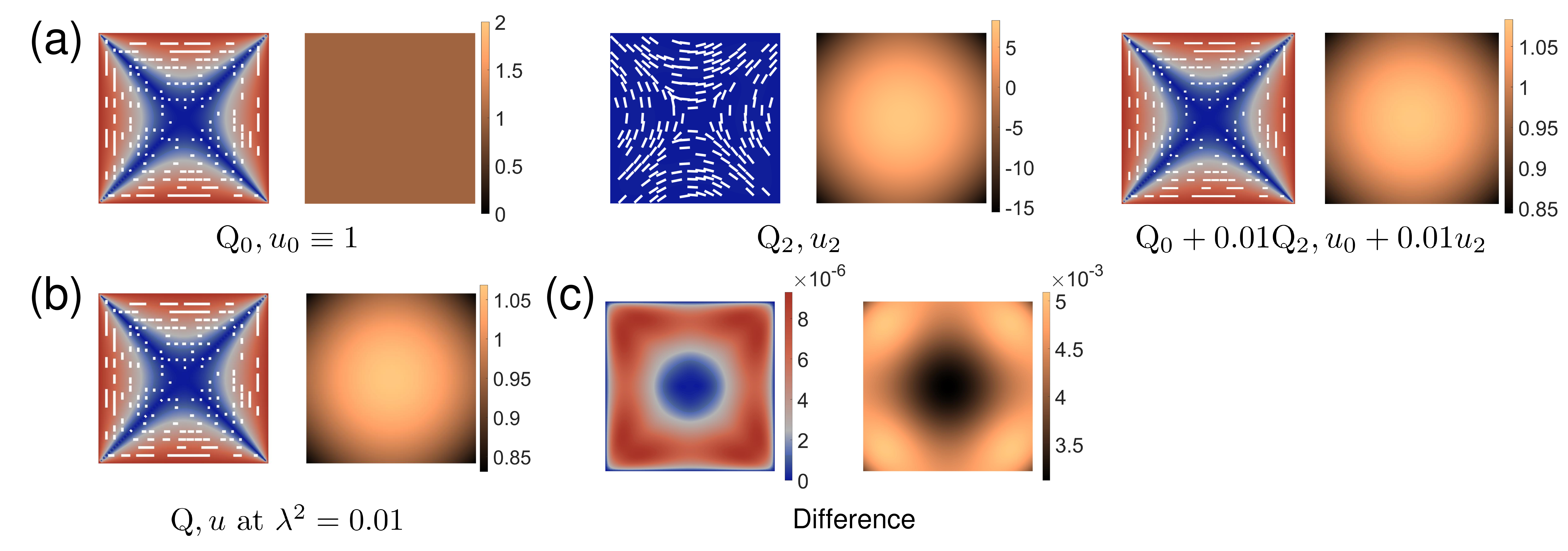}
\caption{(a) From the left to the right are the profiles of $\Q_0=Q_{WORS},u_0\equiv1$, $\Q_2$ and $u_2$ solved from \eqref{second order term equation}, and $\Q_0+0.01\Q_2,u_0+0.01u_2$ 
for $\lambda^2=0.01$. (b) The full numerical solution for $\lambda^2$=0.01. (c) The plot of the difference between $\Q_0+0.01\Q_2,u_0+0.01u_2$ and the numerical solution ($\Q_{\lambda=0.1},u_{\lambda=0.1}$) for $\lambda^2=0.01$, i.e. $\sqrt{\text{Tr}  (\Q_0+0.01\Q_2-\Q_{\lambda=0.1})^2/2}$ and $u_0+0.01u_2-u_{\lambda=0.1}$. The parameters are $a=-5$, $c=5$, $B_0=10^{-3}$, $q=2\pi$, and $A=-0.8359$.}
\label{Asymptotic}
\end{figure}
\subsubsection{Modest domain size} \label{sec:modest}
In this section, we numerically study the stable smectic critical points of \eqref{non-dimensionalised energy} with modest $\lambda$, which complements the $\lambda \rightarrow 0$ and $\lambda \rightarrow \infty$ problems. Unless otherwise specified in the figure caption, the default parameter values are: $a=-5$, $c=5$, $B_0=10^{-3}$, $q=2\pi$ (corresponding to a molecular length of approximately $10^{-7}$ m), and $A = - 0.8359$ (which is calculated from the parameters in \cite{yin2020construction,shi2023hierarchies,robinson2017molecular}). 
For the fixed value of temperature as coded in the values of $a$ and $A$, as the domain size increases from $\lambda^2=1$ to $\lambda^2=30$, three stable smectic states are shown in Figure \ref{configuration as lambda}.  These states are the minimisers of \eqref{non-dimensionalised energy} and have the lowest energy according to our numerical calculations. We can recognise the $\Q$-profiles from the LdG studies: the WORS with two line defects on diagonals, the BD with two line defects localised near opposite edges, and the D state with the nematic director along a square diagonal and with no interior defects \cite{yin2020construction,robinson2017molecular}. The corresponding $u$ profiles have layer normal along the director of $\Q$ profiles. The BD-like and D-like smectic states can be observed in experiments in \cite{cortes2016colloidal}. We note that the BD-like state, which is unstable in the rLdG theory for nematic phase, gains stability in the mLdG framework, at least for some values of $\lambda$. The WORS-like state is hard to achieve practically because of the small domain constraint, which could correspond to nanoscale domains.

\begin{figure}
\centering
\includegraphics[width=.99\textwidth]{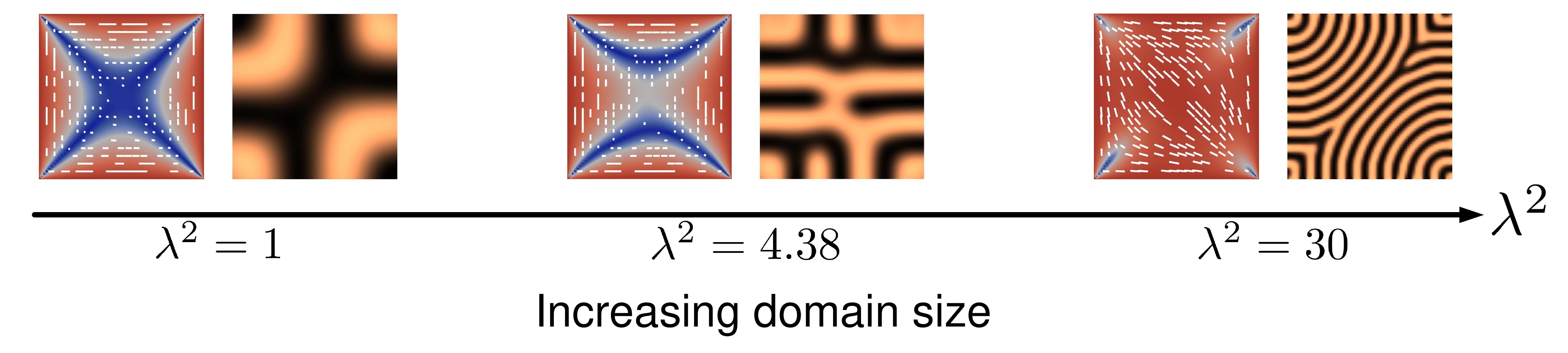}
\caption{From small $\lambda$ to large $\lambda$, the nematic director of minimisers of \eqref{non-dimensionalised energy} exhibit the WORS, BD and D profiles respectively. The colour bar is the same as in Figure \ref{distribution_near_defect}.}
\label{configuration as lambda}
\end{figure}

To further explore the interplay between the positional order and orientational order, we track the solution branches, with small and large $\lambda$, as temperature decreases. In Figure \ref{phase transition confinement}(a), for small $\lambda^2=4.38$, the stable state is the nematic WORS (where $u\equiv0$) for high temperatures. As the temperature decreases, the BD-state appears gradually and separates the cross-line defects into two distinct line defects, localised near a pair of opposite square edges.
We speculate that the stability of the BD-like smectic state is enhanced by the positional order parameter $u$, to avoid more dislocations in the WORS-like smectic state. In Figure \ref{phase transition confinement}(b), for a large domain captured by $\lambda^2=30$, the nematic D state crystallizes into the smectic D-like state, which reflects the memory of the director in the N-S phase transition.

\begin{figure}
\centering
\includegraphics[width=.9\textwidth]{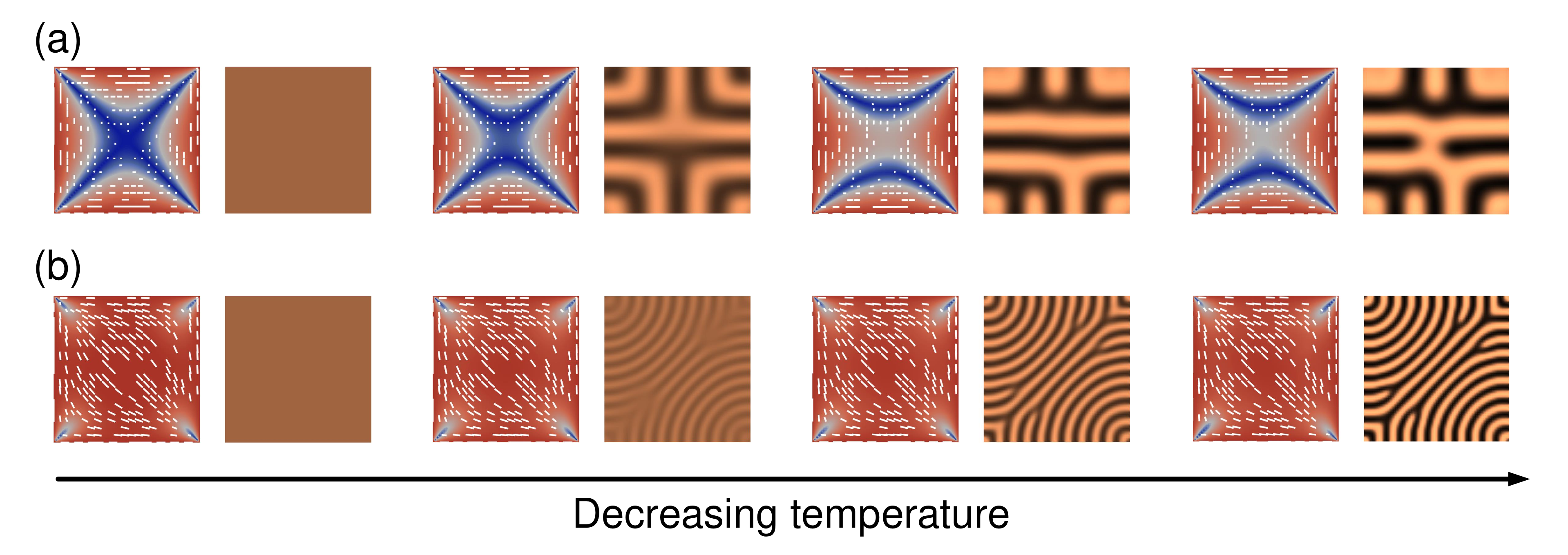}
\caption{
(a) The structural transition from the WORS-type to the BD-type smectic with decreasing temperature for $\lambda^2=4.38$, and the rescaled temperature-dependent parameters are $a=1, -0.2, -2, -5$, $A=-0.4286, -0.5916, -0.7544, -0.8359$ from left to right respectively. (b) Depicts the structural transition between the D-type smectic and the  crystallised D-type smectic with decreasing temperature for $\lambda^2=30$, with the same re-scaled temperature-dependent parameters as in (a). Colour bar as in Figure \ref{distribution_near_defect}.}
\label{phase transition confinement}
\end{figure}

In the LdG theory for nematics, we can find at least six different (meta)stable critical points of the LdG (rLdG) energy on square domains with tangent boundary conditions, when $\lambda$ is large enough. There are two D states, for which the nematic director aligns along one of the square diagonals and four R states, for which the director rotates by $\pi$ radians between two opposite square edges. The profiles of D and R are unique in nematics, once we take symmetry into account \cite{yin2020construction,robinson2017molecular}. However, in the mLdG model, we can find multiple (meta)stable D-like and R-like states with subtle differences in the corresponding $u$ profiles (see Figure \ref{uniaxial stable state}(a)). This could suggest a frustrated energy landscape, implying the existence of numerous similar energy minimisers that differ slightly in their structural details. Intuitively, the $u$-dependent energy densities $f_{bs}(u)$ and $f_{int}(\Q,u)$ are highly nonlinear with respect to $u$ and involve the $L_2$-norm of the second derivative of $u$, which contributes to a frustrated energy landscape. In contrast, the $\Q$-dependent $f_{LdG}(\Q)$ involves only the $L_2$-norm of the first derivative of $\Q$, resulting in a smoother energy landscape \cite{yin2020construction}. We first choose relatively large values for $|a|$, $c$, and $B_0$ to ensure that the frustrated $f_{bs}(u)$ and $f_{int}(\Q,u)$ dominates. By using the saddle dynamics \cite{2019High}, we search for the transition pathway between the R1 and R2 states, via an index-1 transition state R3, in Figure \ref{uniaxial stable state}(b). In such a frustrated energy landscape, it is difficult for an R-like state to break the energy barrier and reach the lower energy D-like states, because the local minima around it are similar R-like states, as shown in Figure \ref{uniaxial stable state}(d). One strategy to alleviate the frustration is to reduce the parameters $|a|, c, B_0$ i.e. make the nematic or LdG energy dominant. By reducing \(|a|\), \(c\), and \(B_0\) in the same ratio, the Euler-Lagrange equation for \(u\) in \eqref{E-L equation} remains unchanged, and the minimizer profiles are not significantly altered. In Figure \ref{uniaxial stable state}(c) and Figure \ref{uniaxial stable state}(e), with reduced parameters $|a|, c, B_0$, the energy landscape is smoother and we find a transition pathway between R-like and D-like states via an index-1 J-like state. This transition pathway is analogous to its purely rLdG counterpart in \cite{yin2020construction,shi2022nematic}. 

\begin{figure}
\centering
\includegraphics[width=.99\textwidth]{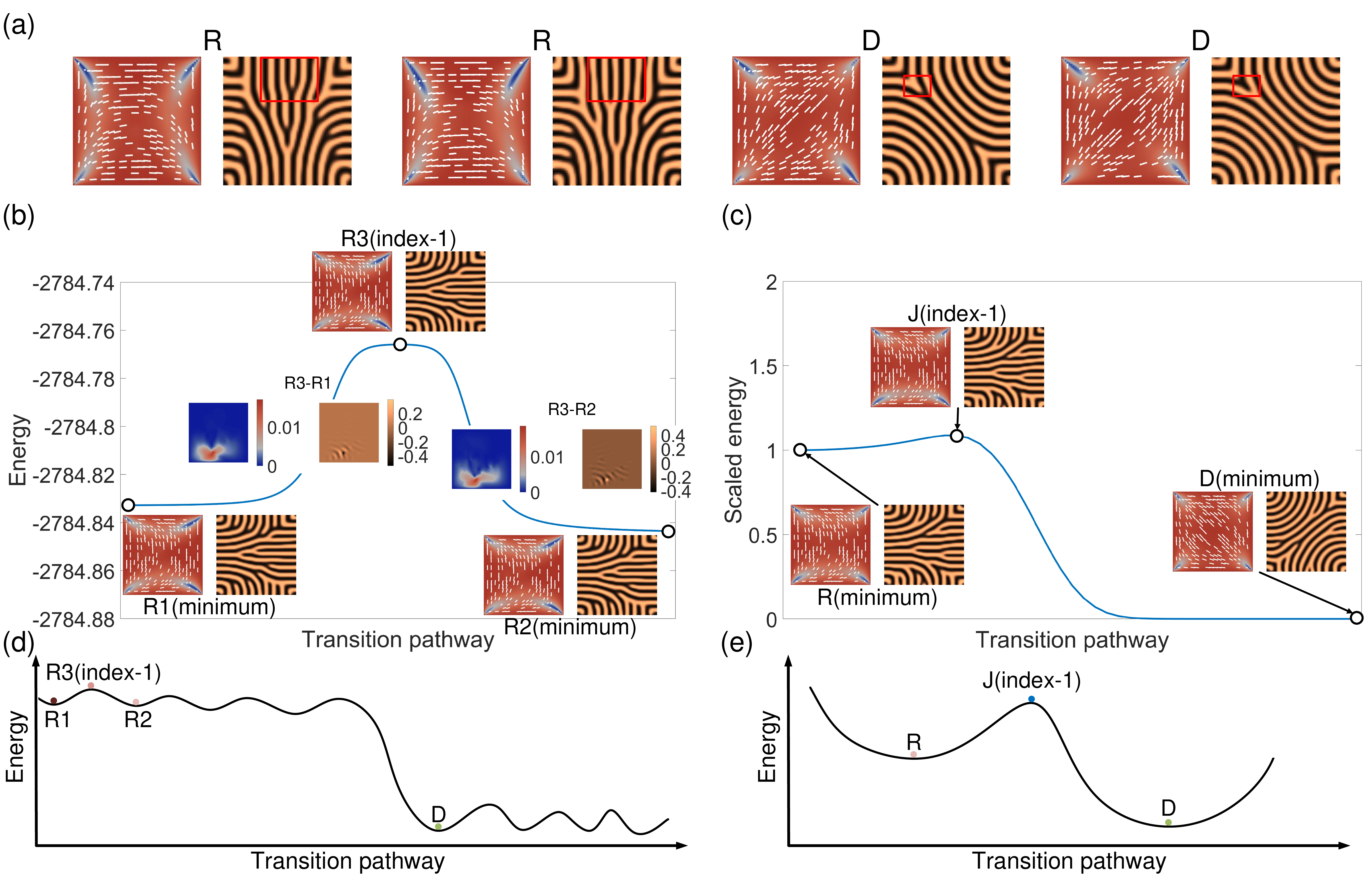}
\caption{(a) R and D type mLdG energy minima for $\lambda^2=30$. The domain enclosed by red lines demonstrates the difference between the two R states and two D states, respectively. (b) A frustrated transition pathway with $\lambda^2=30, B_0=10^{-3},a=-5,c=5$. R1 and R2 are energy minima, and R3 is an index-1 transition state. R3-R1 (R2) is the slight pointwise difference between the R3 and R1 (R2). (c) The transition pathway between locally stable R state and more stable D state via index-1 transition state J with $\lambda^2=30, B_0=10^{-5},a=-0.05,c=0.05$, and the $y$-axis is the scaled energy, $E_{scaled}=e^{E-E(R)}$, for better visualization. The schematics in (d) and (e) represent the frustrated energy landscape in (b) and the smooth energy landscape in (c), respectively.}
\label{uniaxial stable state}
\end{figure}

\section{Conclusion and discussion} \label{conclusion}
We model smectic configurations in the mLdG framework, which is essentially the LdG framework for nematic liquid crystals augmented by a positional/smectic order parameter, $u$, and coupling between the nematic and smectic order parameters. This model was proposed in \cite{xia2021structural} with multiple phenomenological parameters: $a, c$, $B_0$ and $q$. We do various formal calculations to give some physical interpretation of these coefficients, e.g., $a$ should depend on the temperature to capture the N-S phase transition and for sufficiently large domains, the amplitude of the density fluctuations depends on $a$ and $c$, the number of layers is proportional to characteristic geometric parameters, the layer normal is aligned along the nematic director and $q$ is inversely proportional to the smectic layer thickness  and can be interpreted as the layer wave number, at least for mLdG energy minimisers. The smectic layer thickness is often related to typical molecular lengths - the length of the long axis of  a rod-like liquid crystal molecule. Our work allows for a more direct and meaningful comparison with experimental parameters.

More precisely, we first prove the existence and regularity of a minimiser of the mLdG energy in suitable admissible spaces, for three different types of experimentally relevant boundary conditions. Then, we prove that the mLdG energy can model the I-N-S phase transition with respect to temperature. 
We then investigate structural phase transitions on square domains (with edge length $\lambda$) subject to tangent boundary conditions for the nematic $\Q$-tensor. Our primary findings are as follows: (a) in the $\lambda\rightarrow 0$ limit or for (very) small square domains, the mLdG energy minimiser is essentially the nematic WORS without a layer structure; (b) in the $\lambda\rightarrow \infty$ limit or for large square domains, the number of layers increases assuming that $B_0$ and $q$ are independent of temperature and $\lambda$; (c) for a finite but non-zero $\lambda$, the mLdG energy minimisers favor the WORS or BD profiles for small square domains, but prefer to bend the D profiles for large square domains, which is in agreement with experimental results in \cite{cortes2016colloidal}. We find multiple (meta)stable states without interior defects and the transition pathways between them, for large square domains which demonstrates a frustrated energy landscape. 


There are several extensions of this work. 
We plan to generalise our work on square domains to arbitrary 2D polygons, in parallel to the work on polygons in the rLdG/nematic framework carried out in \cite{han2020reduced}, along with generalisations to 3D geometries e.g., cuboid \cite{shi2024multistability} and spherical shells \cite{cortes2016colloidal}. Further, there are limitations of the mLdG model, e.g., the Isotropic-Smectic phase transition \cite{izzo2020landau,biscari2007landau} is outside the scope of the mLdG model. We also plan to develop variants of the mLdG model that can capture multiple phase transitions.

\section{Code availability}
The codes to produce the numerical results of this paper can be obtained
at https://github.com/BaomingShi/searching-minimizer-mLdG-model.

\section*{Acknowledgements}
BS would like to thank the University of Strathclyde for its support and hospitality when work on this paper was undertaken.

\section*{Appendix: Numerical method}\label{numerical method}
A (meta)stable state can be found by the gradient descent method, and a transition state can be found by the saddle dynamic \cite{2019High}. For the confinement problem in Section \ref{Sec: confinement}, we use finite difference methods for spatial discretization with mesh size $\delta x$. The discretization of the gradient flow of \eqref{energy_2D} is,
\begin{equation}
\begin{aligned}
&\frac{\Q_{n+1}-\Q_n}{\Delta t_n}=-K\Delta_{\delta x} \Q_n - A\cdot \Q_n-C\cdot tr(\Q_n^2)\Q_n\\
&-2B_0\cdot q^2/s_+\cdot\left(u_n\cdot D^2_{\delta x} u_n-\frac{tr(u\cdot D^2_{\delta x} u_n)}{2}\mathbf{I}_2\right)-2\cdot B_0\cdot q^4\cdot \frac{\Q_n}{s_+^2}u_n^2,\\
&\frac{u_{n+1}-u_n}{\Delta t_n}=-2B_0\Delta_{\delta x}^2 u_{n+1} - a u_n-c u_n^3-2B_0\cdot D^2_{\delta x}u_n:\left(q^2\cdot\left(\frac{\Q_n}{s_+}+\frac{\mathbf{I}_2}{2}\right)\right)\\
&-2B_0\cdot \nabla_{\delta x}\cdot \left(\nabla_{\delta x} \cdot\left(q^2\cdot\left(\frac{\Q_n}{s_+}+\frac{\mathbf{I}_2}{2}\right)u_n\right)\right)-2B_0 \cdot \left|q^2\cdot\left(\frac{\Q_n}{s_+}+\frac{\mathbf{I}_2}{2}\right)\right|^2 u_n,
\end{aligned}
\label{discretization}
\end{equation}
where $\Delta_{\delta x}^2, \Delta_{\delta x}, \nabla_{\delta x}$, $D^2_{\delta x}$ are the discretization of  $\Delta^2, \Delta, \nabla, D^2$, and $\Delta t_n$ is the Barzilai-Borwein (BB) step size \cite{BB} at the $n$-th iteration. In \eqref{discretization}, we discretize the fourth-order operator $\Delta^2$ implicitly to ensure the stability of the BB step size. 

In Section \ref{phase transition}, we study the phase transition problem with periodic boundary conditions, and we use the spectral method \cite{shen2011spectral} for spatial discretization,
\begin{equation}
    Q(x)=\sum_{k=-N/2}^{N/2} \Tilde{Q}_k e^{2\pi i kx/h}, Q\in V_Q,
        u(x)=\begin{cases}
        \sum_{k=-N/2}^{N/2} \Tilde{u}_k e^{2\pi i kx/h}, u\in V_u,\\
        \sum_{k=1}^{N+1} \Tilde{u}_k \sin\left(2k\pi x/h \right), u\in V,
    \end{cases}
    \label{spectral method for u}
\end{equation}
where $N$ is an even integer, and we choose $N=32$. Recall that $V=V_u \cap W^{1,2}_{0,\Omega}$, so we use the sine spectral method to discretize $u\in V$. By substituting \eqref{spectral method for u} in \eqref{energy:simple}, we obtain a discretized form of the energy,
\begin{equation}
    E(\Tilde{Q}_k,\Tilde{u}_k)\approx E(Q,u).
\end{equation}
This results in a function of $2(N+1)$ variables, and we directly search for the minimum by using the gradient descent method for finite-dimensional functions.

\bibliographystyle{unsrt}
\bibliography{references}
\end{document}